\documentclass[12pt,oneside,english]{amsart}
\usepackage{lmodern}

\usepackage[T1]{fontenc}
\usepackage[latin9]{inputenc}
\usepackage{geometry}
\usepackage{color}
\usepackage{babel}
\usepackage{rotfloat}
\usepackage{booktabs}
\usepackage{amstext}
\usepackage{amsthm}
\usepackage{amssymb}
\usepackage{graphicx}
\usepackage{setspace}
\usepackage[authoryear]{natbib}
\usepackage[unicode=true,pdfusetitle,
 bookmarks=true,bookmarksnumbered=false,bookmarksopen=false,
 breaklinks=false,pdfborder={0 0 0},pdfborderstyle={},backref=false,colorlinks=true]
 {hyperref}
\hypersetup{
 citecolor=blue, linkcolor=blue}
\usepackage{breakurl}

\makeatletter

\providecommand{\tabularnewline}{\\}
\newcommand{\lyxdot}{.}

\theoremstyle{plain}
\newtheorem{thm}{\protect\theoremname}
  \theoremstyle{remark}
  \newtheorem{rem}{\protect\remarkname}
  \theoremstyle{plain}
  \newtheorem{prop}{\protect\propositionname}
  \theoremstyle{plain}
  \newtheorem{cor}{\protect\corollaryname}
  \theoremstyle{plain}
  \newtheorem{lem}{\protect\lemmaname}

\usepackage{calc}\usepackage{fixltx2e}\usepackage{multicol}\usepackage[normalem]{ulem}
\usepackage{ae,aecompl}
\usepackage[T1]{fontenc}
\usepackage[latin9]{inputenc}
\usepackage{geometry}
\usepackage{array}
\usepackage{float}
\usepackage{multirow}
\usepackage{amstext}
\usepackage{amsthm}
\usepackage{amssymb}
\usepackage{graphicx}
\usepackage{setspace}
\usepackage{esint}
\usepackage{graphicx,amsmath,amssymb,epsfig}
\usepackage{amsfonts}
\usepackage{geometry}
\usepackage{setspace}
\usepackage{lscape}
\usepackage{caption}
\usepackage{rotating}
\usepackage{xcolor}
\usepackage{amsfonts}
\usepackage{subfig}
\usepackage{babel}

\theoremstyle{definition}\newtheorem{assumption}{Assumption}
\providecommand{\tabularnewline}{\\}
\makeatletter
\floatstyle{ruled}
\numberwithin{equation}{section}
\numberwithin{figure}{section}
\theoremstyle{remark}
  
\date{\today. We are grateful to Frank Diebold, James MacKinnon, Alex Tetenov, Daniel Wilhelm, 
participants to the Bristol Econometric Study Group 2018 
and the 13th GNYMA Econometrics Colloquium, Princeton, 
for useful comments and discussions, 
and to Richard Blundell, Joel Horowitz and Matthias Parey 
for sharing the data used in the demand for gasoline empirical illustration. The Supplementary Material is available at \href{https://samistouli.com/}{https://samistouli.com/}.}
\thanks{\noindent $^\dag$ Nuffield College, Oxford, and Department of Economics, Johns Hopkins University, 
rspady@jhu.edu.}
\thanks{ $\S$ Department of Economics, University of Bristol, 
s.stouli@bristol.ac.uk.}

\@ifundefined{showcaptionsetup}{}{%
 \PassOptionsToPackage{caption=false}{subfig}}
\usepackage{subfig}
\makeatother

  \providecommand{\lemmaname}{Lemma}
  \providecommand{\propositionname}{Proposition}
  \providecommand{\remarkname}{Remark}
\providecommand{\corollaryname}{Corollary}
\providecommand{\theoremname}{Theorem}

\begin{document}

\title{Simultaneous Mean-Variance Regression}

\author{Richard H. Spady$^\dag$ and Sami Stouli$^\S$}
\begin{abstract}
We propose simultaneous mean-variance regression for the linear estimation
and approximation of conditional mean functions. In the presence of
heteroskedasticity of unknown form, our method accounts for varying
dispersion in the regression outcome across the support of conditioning
variables by using weights that are jointly determined with the mean
regression parameters. Simultaneity generates outcome predictions
that are guaranteed to improve over ordinary least-squares prediction
error, with corresponding parameter standard errors that are automatically
valid. Under shape misspecification of the conditional mean and variance
functions, we establish existence and uniqueness of the resulting
approximations and characterize their formal interpretation and robustness
properties. In particular, we show that the corresponding mean-variance
regression location-scale model weakly dominates the ordinary least-squares
location model under a Kullback-Leibler measure of divergence, with
strict improvement in the presence of heteroskedasticity. The simultaneous
mean-variance regression loss function is globally convex and the
corresponding estimator is easy to implement. We establish its consistency
and asymptotic normality under misspecification, provide robust inference
methods, and present numerical simulations that show large improvements
over ordinary and weighted least-squares in terms of\textbf{ }estimation
and inference in finite samples. We further  illustrate our method
with two empirical applications to the estimation of the relationship
between economic prosperity in 1500 and today, and demand for gasoline
in the United States.
\end{abstract}

\maketitle

\section{Introduction}

Ordinary least-squares (OLS) is the method of choice for the linear
estimation and approximation of the conditional mean function (CMF).
However, in the presence of heteroskedasticity the standard errors
of OLS are inconsistent, and subsequent inference is therefore unreliable.
As a way of achieving valid inference, practitioners instead often
use the heteroskedasticity-corrected standard errors of Eicker (\citeyear{Eicker:1963},
\citeyear{Eicker:1967}), \citet{Huber:1967} and \citet{White:1980a}.
Although valid asymptotically, numerous limitations of this approach
have been highlighted in the literature such as bias and sensitivity
to outliers, incorrect size and low power of robust tests in finite
samples (\citealp{MKW:1985}; \citealp{CJ:1987}; \citealp{Chesher:1989};
\citealp{CA1991}). These findings in turn generated a large number
of proposals in order to reconcile the large-sample validity of the
approach and its observed finite-sample limitations, surveyed in \citet{MacKinnon:2013}.

The finite-sample limitations of OLS-based inference essentially originate
from the fact that OLS assigns a constant weight to each observation
in fitting the best linear predictor for the regression outcome. Hence
the least-squares criterion does not account for the varying accuracy
of the information available about the outcome across the covariate
space. This yields point estimates and linear approximations that
are sensitive to high-leverage points and outliers, which in turn
generate biased estimates of the residuals' second moments used in
the calculation of the robust variance-covariance matrix of OLS parameters.
In finite samples, uniform weighting not only compromises the validity
of OLS-based statistical inference in the presence of heteroskedasticity,
but also the reliability of OLS point estimates.

In this paper, we propose simultaneous mean-variance regression (MVR)
as an alternative to OLS for the linear estimation and approximation
of CMFs. MVR characterizes the conditional mean and variance functions
jointly, thereby providing a solution to the problems of estimation,
approximation and inference in the presence of heteroskedasticity
of unknown form with five main features. First, it incorporates information
from the second conditional moment in the determination of the first
conditional moment parameters. Second, simultaneity generates approximations
with improved robustness properties relative to OLS. Third, the resulting
approximations have a formal interpretation under shape misspecification
of the conditional mean and variance functions. Fourth, MVR solutions
are well-defined, i.e., exist and are unique. Fifth, standard errors
of the corresponding estimator are automatically valid in the presence
of heteroskedasticity of unknown form, and reduce to those of OLS
under homoskedasticity.

The MVR criterion can be interpreted as a penalized weighted least-squares
(WLS) loss function. The presence and the form of the penalty ensure
global convexity of the objective function, so that MVR conditional
mean and variance approximations are jointly well-defined. This differs
from the usual WLS approach where a sequential procedure is followed,
obtaining the weights first, and then implementing a weighted regression
to determine the parameters of the linear specification. Our simultaneous
approach allows us to give theoretical guarantees on the relative
approximation properties of MVR and OLS. We use MVR to construct and
estimate a new class of approximations of the conditional mean and
variance functions, with improved robustness and precision in finite
samples. We establish the interpretation of MVR approximations, we
derive the asymptotic properties of the corresponding MVR estimator,
and we give tools for robust inference. We also illustrate the practical
benefits of the MVR estimator with extensive numerical simulations,
and find very large finite-sample improvements over both OLS and WLS
in terms of estimation performance and heteroskedasticity-robust inference.

This paper generalizes the results of \citet{SS:2018} for the primal
problem of the dual regression estimator of linear location-scale
models. We provide a unified theory allowing for a large class of
scale functions. This paper is also related to the interpretation
of OLS under misspecification of the shape of the CMF. OLS gives the
minimum mean squared error linear approximation to the CMF, an important
motivation for its use in empirical work (\citealp{White:1980b};
\citealp{Chamberlain:1984}; \citealp{Angrist:Krueger:1999}; \citealp{AP:2008}).
MVR introduces a class of WLS approximations accounting for potential
variation in the outcome across the support of conditioning variables,
and with weights that have a clear interpretation under misspecification.
Our approach thus complements the textbook WLS proposal of \citet{CT:2005}
and Wooldridge (\citeyear{Wooldridge:2010}, \citeyear{Wooldridge:2012})
(see also \citealp{RW:2017}), who advocate the reweighting of OLS
with generalized least-squares (GLS) weights and further correcting
the standard errors for heteroskedasticity.

This paper makes three main contributions.\textbf{ }First, we show
existence and uniqueness of MVR solutions under general misspecification,
thereby introducing a new class of location-scale models corresponding
to MVR approximations. The results in \citet{SS:2018} did not cover
the case of misspecified conditional mean and variance functions.
Second, we establish favorable approximation and robustness properties
of MVR relative to OLS. We show that MVR is a minimum WLS linear approximation
to the CMF, with weights determined such that the MVR approximation
improves over OLS in the presence of heteroskedasticity under the
MVR loss. For our main specifications of the scale function, we further
show that OLS root mean squared prediction error is an upper bound
for the MVR weighted mean squared prediction error. We then extend
this result to show that under a Kullback-Leibler information criterion
(KLIC) the proposed MVR location-scale models weakly dominate the
OLS location model, with strict improvement in the presence of heteroskedasticity.
These results provide theoretical guarantees motivating the use of
MVR over OLS, and are not shared by alternative WLS proposals. Third,
we derive the asymptotic distribution of the MVR estimator under misspecification
and provide robust inference methods. In particular we propose a robust
one-step heteroskedasticity test that complements existing OLS-based
tests (e.g., \citealp{BP:1979}; \citealp{White:1980a}; \citealp{Koenker:1981}).

The rest of the paper is organized as follows. Section \ref{sec:SMVR}
introduces MVR under correct specification of conditional mean and
variance functions. Section \ref{sec:Approximation-Properties} establishes
the main approximation properties of MVR under misspecification, including
existence and uniqueness. Section \ref{sec:Estimation-and-Inference}
gives asymptotic theory. Section \ref{sec:Numerical-Illustrations}
reports the results of an empirical application to the relationship
between economic prosperity in 1500 and today, and illustrates the
finite-sample performance of MVR with numerical simulations. All proofs
of the main results are given in the Appendix. The online Appendix
Spady and Stouli (2018) contains supplemental material, including
an additional empirical application to demand for gasoline in the
United States.

\section{Simultaneous Mean-Variance Regression\label{sec:SMVR}}

\subsection{The Mean-Variance Regression Problem\label{subsec:MVR problem}}

Given a scalar random variable $Y$ and a random $k\times1$ vector
$X$ that includes an intercept, i.e., has first component $1$, denote
the mean and standard deviation functions of $Y$ conditional on $X$
by $\mu(X):=E[Y\mid X]$ and $\sigma(X):=E[(Y-E[Y\mid X])^{2}\mid X]^{1/2}$,
respectively. We start with a simplified setting where the conditional
mean and variance functions take the parametric forms
\begin{equation}
\mu(X)=X'\beta_{0},\quad\quad\sigma(X)^{2}=s(X'\gamma_{0})^{2},\label{eq:Model}
\end{equation}
for some positive scale function $t\mapsto s(t)$, and where the parameters
$\beta_{0}$ and $\gamma_{0}$ belong to the parameter space $\Theta=\mathbb{R}^{k}\times\Theta_{\gamma}$,
with $\Theta_{\gamma}=\{\gamma\in\mathbb{R}^{k}\,:\,\Pr[s(X'\gamma)>0]=1\}$.
Two leading examples for the scale function are the linear and exponential
specifications $s(t)=t$ and $s(t)=\exp(t)$, with domains $(0,\infty)$
and $\mathbb{R}$, respectively.

The parameter vector $\theta_{0}:=(\beta_{0},\gamma_{0})'$ is uniquely
determined as the solution to the globally convex MVR population problem
\begin{equation}
\min_{\theta\in\Theta}E\left[\frac{1}{2}\left\{ \left(\frac{Y-X'\beta}{s(X'\gamma)}\right)^{2}+1\right\} s(X'\gamma)\right].\label{eq:MVR}
\end{equation}
When the functions $x\mapsto\mu(x)$ and $x\mapsto\sigma(x)^{2}$
satisfy model (\ref{eq:Model}), they are simultaneously characterized
by problem (\ref{eq:MVR}). As a consequence, MVR incorporates information
on the dispersion of $Y$ across the support of $X$ in the determination
of the mean parameter $\beta$. We show below that problem (\ref{eq:MVR})
is formally equivalent to an infeasible sequential least-squares estimator
of the conditional mean and variance functions for model (\ref{eq:Model}).
Problem (\ref{eq:MVR}) is a generalization of the dual regression
primal problem introduced in \citet{SS:2018}, for which the scale
function is linear. Considering scale functions with domain the real
line, such as the exponential function, allows the transformation
of the dual regression primal problem into an unconstrained convex
problem over $\Theta=\mathbb{R}^{2\times k}$.

Inspection of the first-order conditions confirms that $\theta_{0}$
is indeed a valid solution to problem (\ref{eq:MVR}). Denoting the
derivative of the scale function by $s_{1}(t):=\partial s(t)/\partial t$
and letting $e\left(Y,X,\theta\right):=(Y-X'\beta)/s(X'\gamma)$,
the first-order conditions of (\ref{eq:MVR}) are
\begin{align}
E\left[Xe\left(Y,X,\theta\right)\right] & =0\label{eq:FOC1}\\
E\left[Xs_{1}(X'\gamma)\{e\left(Y,X,\theta\right)^{2}-1\}\right] & =0.\label{eq:FOC2}
\end{align}
These conditions are satisfied by $\theta_{0}$ since specification
(\ref{eq:Model}) is equivalent to the location-scale model
\begin{equation}
Y=X'\beta_{0}+s(X'\gamma_{0})\varepsilon,\quad E[\varepsilon\mid X]=0,\quad E[\varepsilon^{2}\mid X]=1.\label{eq:loc-scale model}
\end{equation}
Therefore, the parameter vector $\theta_{0}$ also satisfies the relations
\begin{align*}
E\left[e(Y,X,\theta_{0})\mid X\right]=E\left[\varepsilon\mid X\right] & =0\\
E\left[e(Y,X,\theta_{0})^{2}-1\mid X\right]=E\left[\varepsilon^{2}-1\mid X\right] & =0,
\end{align*}
which imply that $E[h(X)e(Y,X,\theta_{0})]=0$ and $E[h(X)\{e(Y,X,\theta_{0})^{2}-1\}]=0$
hold for any measurable function $x\mapsto h(x)$, and in particular
for $h(X)=X$ and $h(X)=Xs_{1}(X'\gamma)$.

\subsection{Formal Framework}

Let $\mathcal{X}$ denote the support of $X$, and for a vector $u=(u_{1},\ldots,u_{k})'\in\mathbb{R}^{k}$,
let $||\cdot||$ denote the Euclidean norm, i.e., $||u||=(u_{1}^{2}+\ldots+u_{k}^{2})^{1/2}$;
we define a compact subset $\Theta^{c}\subset\Theta$ as
\[
\Theta^{c}:=\left\{ \theta\in\Theta\,:\,||\theta||\leq C_{\theta}\:\textrm{and}\:\inf_{x\in\mathcal{X}}s(x'\gamma)\geq C_{s}\right\} ,
\]
for some finite constant $C_{\theta}$ and some constant $C_{s}>0$,
with interior set denoted $\textrm{int}(\Theta^{c})$. The second
and third derivatives of the scale function $t\mapsto s(t)$ are denoted
by $s_{j}(t):=\partial^{j}s(t)/\partial t^{j}$, $j=2,3$. We also
denote the MVR objective function in (\ref{eq:MVR}) by $Q(\theta):=E[\{e\left(Y,X,\theta\right)^{2}+1\}s(X'\gamma)/2]$. 

Our first assumption specifies the class of scale functions we consider.

\begin{assumption}For $a=0$ or $-\infty$, the scale function $s:(a,\infty)\rightarrow(0,\infty)$
is a three times differentiable strictly increasing convex function
that satisfies $\lim_{t\rightarrow a}s(t)=0$ and $\lim_{t\rightarrow\infty}s(t)=\infty$.\label{ass:ScaleSpec}\end{assumption}

Assumption \ref{ass:ScaleSpec} encompasses several types of scale
functions such as polynomial specifications $s(x'\gamma)=(x'\gamma)^{\alpha}$
with $a=0$ and $\Pr[X'\gamma>0]=1$, or exponential-polynomial specifications
$s(x'\gamma)=\exp(x'\gamma)^{\alpha}$ with $a=-\infty$, for some
$\alpha>0$. For $\alpha=1$, we recover the linear and exponential
scale leading cases. 

The next assumptions complete our formal framework.

\begin{assumption}The conditional variance function $x\mapsto\sigma(x)^{2}$
is bounded away from $0$ uniformly in $\mathcal{X}$.\label{ass:Var(Y|X)}\end{assumption}

\begin{assumption}We have (i) $E[Y^{4}]<\infty$ and $E||X||^{4}<\infty$,
and, (ii) for all $\gamma\in\Theta_{\gamma}$, $E[\left\Vert X\right\Vert ^{4}s_{2}(X'\gamma)^{2}]<\infty$,
$E[\left\Vert X\right\Vert ^{6}s_{3}(X'\gamma)^{2}]<\infty$ and $E[\left\Vert X\right\Vert ^{6}s_{1}(X'\gamma)^{2}s_{2}(X'\gamma)^{2}]<\infty$.\label{ass:Moments}\end{assumption}

\begin{assumption}For all $\gamma\in\Theta_{\gamma}$, $E[XX'/s(X'\gamma)]$
is nonsingular.\label{ass:FullRank}\end{assumption}

Assumptions \ref{ass:ScaleSpec}-\ref{ass:FullRank} are sufficient
conditions for global convexity of the MVR criterion over the parameter
space $\Theta$, and therefore for problem (\ref{eq:MVR}) to have
a unique solution.
\begin{thm}
\label{thm:Uniqueness}If Assumptions \ref{ass:ScaleSpec}-\ref{ass:FullRank}
hold, and the conditional mean and variance functions of $Y$ given
$X$ satisfy model (\ref{eq:Model}) a.s. with $\theta_{0}\in\textrm{int}(\Theta^{c})$,
then $\theta_{0}$ is the unique minimizer of $Q(\theta)$ over $\Theta$.
\end{thm}
Theorem \ref{thm:Uniqueness} applies when the conditional mean and
variance functions are well-specified, and thus provides primitive
conditions for identification of $\theta_{0}$ in the location-scale
model (\ref{eq:loc-scale model}). This extends the uniqueness result
in \citet{SS:2018} for objective (\ref{eq:MVR}) with a linear scale
function to the class of scale functions defined in Assumption \ref{ass:ScaleSpec}.
\begin{rem}
In the linear scale case, $s_{1}(t)=1$ and $s_{j}(t)=0$, $j=2,3$,
so that Assumption \ref{ass:Moments} reduces to Assumption \ref{ass:Moments}(i).
In the exponential scale case, $s_{j}(t)=\exp(t)$, $j=1,2,3$, so
that Assumption \ref{ass:Moments}(ii) reduces to the requirement
that $E[\left\Vert X\right\Vert ^{6}\exp(X'\gamma)^{4}]$ be finite.
This is satisfied for instance if $X$ is bounded.\qed
\end{rem}

\subsection{Simultaneous Mean-Variance Regression Interpretation}

Problem (\ref{eq:MVR}) is equivalent to an infeasible sequential
least-squares estimator of conditional mean and variance functions.
The first-order conditions of (\ref{eq:MVR}) can also be written
as
\begin{eqnarray}
E\left[\frac{X}{s(X'\gamma)}\left(Y-X'\beta\right)\right] & = & 0\label{eq:FOC1-2}\\
E\left[X\frac{s_{1}(X'\gamma)}{s(X'\gamma)^{2}}\left\{ (Y-X'\beta)^{2}-s(X'\gamma)^{2}\right\} \right] & = & 0.\label{eq:FOC2-2}
\end{eqnarray}
Given knowledge of $\gamma_{0}$, WLS regression of $Y$ on $X$ with
weights $1/s(X'\gamma_{0})$ has first-order conditions (\ref{eq:FOC1-2}),
with solution $\beta_{0}$. Moreover, given knowledge of $\beta_{0}$,
nonlinear WLS regression of $(Y-X'\beta_{0})^{2}$ on $X$ with weights
$1/s(X'\gamma_{0})^{3}$ and quadratic link function has first-order
conditions (\ref{eq:FOC2-2}), and therefore solution $\gamma_{0}$.
\begin{prop}
\label{prop:Equivalence}If Assumptions \ref{ass:ScaleSpec}-\ref{ass:FullRank}
hold, and (i) $E[Y^{4}]<\infty$, $E[\left\Vert X\right\Vert ^{4}]<\infty$
and $E[s(X'\gamma)^{4}]<\infty$ for all $\gamma\in\Theta_{\gamma}$,
and (ii) the conditional mean and variance functions of $Y$ given
$X$ satisfy model (\ref{eq:Model}) a.s., then the MVR population
problem (\ref{eq:MVR}) is equivalent to the infeasible sequential
estimator with first step
\begin{equation}
\beta_{0}=\arg\min_{\beta\in\Theta_{\beta}}E\left[\frac{1}{\sigma(X)}(Y-X'\beta)^{2}\right],\label{eq:InfeasibleWLS}
\end{equation}
and second step
\begin{equation}
\gamma_{0}=\arg\min_{\gamma\in\Theta_{\gamma}}E\left[\frac{1}{\sigma(X)^{3}}\left\{ (Y-X'\beta_{0})^{2}-s(X'\gamma)^{2}\right\} ^{2}\right].\label{eq:InfeasibleWNLS}
\end{equation}
\end{prop}
An immediate implication of the Law of Iterated Expectations and Proposition
\ref{prop:Equivalence} is that MVR implements simultaneous weighted
linear regression of $\mu(X)$ on $X$ and weighted nonlinear regression
of $\sigma(X)^{2}$ on $X$ by solving for $\beta$ and $\gamma$
such that the weighted residuals $(\mu(X)-X'\beta)/s(X'\gamma)$ and
$\{\sigma(X)^{2}-s(X'\gamma)^{2}\}/s(X'\gamma)^{2}$ are simultaneously
orthogonal to $X$ and $Xs_{1}(X'\gamma)$, respectively. Proposition
\ref{prop:Equivalence} thus establishes the simultaneous mean and
variance regression interpretation of problem (\ref{eq:MVR}).

\section{Approximation Properties of MVR under Misspecification\label{sec:Approximation-Properties}}

Under misspecification, OLS provides the minimum mean squared error
linear approximation to the CMF. For the proposed MVR criterion, existence
of an approximating solution and the nature of the approximation are
nontrivial when the shapes of the conditional mean and variance functions
are misspecified. In this section, we first establish existence and
uniqueness of a solution to the MVR problem under misspecification,
and then characterize the interpretation and properties of the corresponding
MVR approximations.

\subsection{Existence and Uniqueness of an MVR Solution}

Assumptions \ref{ass:ScaleSpec}-\ref{ass:FullRank} are sufficient
for characterizing the smoothness properties, shape, and behaviour
on the boundaries of the parameter space of the MVR criterion $Q(\theta)$.
Under these assumptions $\theta\mapsto Q(\theta)$\textbf{ }is continuous
and its level sets are compact. Compactness of the level sets is a
sufficient condition for existence of a minimizer in $\Theta$, and
is a consequence of the explosive behaviour of the objective function
at the boundaries of the parameter space. The objective $Q(\theta)$
is a \textit{coercive} function over the open set $\Theta$, i.e.,
it satisfies 
\[
\lim_{||\theta||\rightarrow\infty}Q(\theta)=\infty,\quad\quad\lim_{\theta\rightarrow\partial\Theta}Q(\theta)=\infty,
\]
where $\partial\Theta$ is the boundary set of $\Theta$. Thus the
MVR criterion is infinity at infinity, and for any sequence of parameter
values in $\Theta$ approaching the boundary set $\partial\Theta$,
the value of the objective is also driven towards infinity. Therefore,
the level sets of the objective function have no limit point on their
boundary, ruling out existence of a boundary solution, and continuity
of $\theta\mapsto Q(\theta)$ is then sufficient to conclude that
it admits a minimizer. Continuity and coercivity of the objective
function are the two properties that guarantee existence of at least
one minimizer in $\Theta$.\footnote{The boundary set of $\Theta$ may be empty, for instance for the exponential
scale specification. In that case the coercivity property reduces
to $\lim_{||\theta||\rightarrow\infty}Q(\theta)=\infty$.} Assumptions \ref{ass:ScaleSpec}-\ref{ass:FullRank} are also sufficient
for $\theta\mapsto Q(\theta)$ to be strictly convex, and therefore
further ensure that $Q(\theta)$ admits at most one minimizer in $\Theta$.
\begin{thm}
\label{thm:Existence}If Assumptions \ref{ass:ScaleSpec}-\ref{ass:FullRank}
hold, then there exists a unique solution $\theta^{*}\in\Theta$ to
the MVR population problem (\ref{eq:MVR}).
\end{thm}
Theorem \ref{thm:Existence} is the second main result of the paper.
It establishes that the MVR problem (\ref{eq:MVR}) has a well-defined
solution, and an immediate corollary is the existence and uniqueness
of the MVR location-scale representation
\[
Y=X'\beta^{*}+s(X'\gamma^{*})e,\quad E[Xe]=0,\quad E[Xs_{1}(X'\gamma^{*})(e^{2}-1)]=0.
\]
This result clarifies further how MVR generalizes OLS by establishing
the existence and the form of the MVR location-scale model when no
shape restrictions are imposed on the conditional mean and variance
functions. The OLS location model is a particular case with the scale
function restricted to be a constant function. 

\begin{sloppy}Although a unique MVR approximation exists irrespective
of the nature of the misspecification, the interpretation of the MVR
approximating functions $x\mapsto(x'\beta^{*},s(x'\gamma^{*})^{2})$
depends on which of the conditional moment functions is misspecified.
We distinguish two types of shape misspecification:
\begin{enumerate}
\item Mean misspecification: the CMF $x\mapsto\mu(x)$ is misspecified.
\item Variance misspecification: only the conditional variance function
$x\mapsto\sigma(x)^{2}$ is misspecified.
\end{enumerate}
The case when both the conditional mean and variance functions are
misspecified is a particular case of mean misspecification.\par\end{sloppy}

\subsection{Interpretation Under Mean Misspecification}

The location-scale representation
\begin{equation}
Y=\mu(X)+\sigma(X)\varepsilon,\quad E[\varepsilon\mid X]=0,\quad E[\varepsilon^{2}\mid X]=1,\label{eq:loc-scale rep}
\end{equation}
provides a general expression for $Y$ in terms of its conditional
mean and standard deviation functions, and is always valid, as long
as first and second conditional moments exist. Substituting expression
(\ref{eq:loc-scale rep}) for $Y$ into the MVR objective function
$Q(\theta)$ gives rise to a criterion for the joint approximation
of $x\mapsto(\mu(x),\sigma(x)^{2})$.

The criterion $Q(\theta)$ can also be appropriately restricted in
order to define the corresponding OLS approximations. Letting $\Theta_{\gamma,\textrm{LS}}=\{\gamma\in\mathbb{R}\,:\,s(\gamma)>0\}$,
define $\Theta_{\textrm{LS}}=\mathbb{R}^{k}\times\Theta_{\gamma,\textrm{LS}}$.
Upon setting $s(X'\gamma)=s(\gamma)$ in the MVR problem (\ref{eq:MVR}),
\[
(\beta_{\textrm{LS}},\gamma_{\textrm{LS}}):=\arg\min_{(\beta,\gamma)\in\Theta_{\textrm{LS}}}E\left[\frac{1}{2}\left\{ \left(\frac{Y-X'\beta}{s(\gamma)}\right)^{2}+1\right\} s(\gamma)\right]
\]
is a particular case of MVR. Since the OLS solution $\theta_{\textrm{LS}}:=(\beta_{\textrm{LS}},\gamma_{\textrm{LS}},0_{k-1})'$
belongs to the parameter space $\Theta$, uniqueness of $\theta^{*}$
implies that the OLS approximation of the conditional moment functions
$x\mapsto(\mu(x),\sigma(x)^{2})$ cannot improve upon the MVR approximation,
according to the MVR loss.
\begin{thm}
\label{thm:Interpretation}If Assumptions \ref{ass:ScaleSpec}-\ref{ass:FullRank}
hold, then the MVR population problem (\ref{eq:MVR}) has the following
properties.

(i) Problem (\ref{eq:MVR}) is equivalent to the infeasible problem
\begin{align}
\min_{\theta\in\Theta} & \;\frac{1}{2}E\left[\left\{ \left(\frac{\mu(X)-X'\beta}{s(X'\gamma)}\right)^{2}+1\right\} s(X'\gamma)\right]+\frac{1}{2}E\left[\frac{\sigma(X)^{2}}{s(X'\gamma)}\right],\label{eq:theta*}
\end{align}
with first-order conditions
\begin{eqnarray}
E\left[X\left(\frac{\mu(X)-X'\beta}{s(X'\gamma)}\right)\right] & = & 0\label{eq:FOCmiss1}\\
E\left[Xs_{1}(X'\gamma)\left\{ \left(\frac{\mu(X)-X'\beta}{s(X'\gamma)}\right)^{2}+\left(\frac{\sigma(X)^{2}}{s(X'\gamma){}^{2}}-1\right)\right\} \right] & = & 0.\label{eq:FOCmiss2}
\end{eqnarray}
(ii) The optimal value of problem (\ref{eq:MVR}) satisfies $Q(\theta^{*})\leq Q(\theta_{\textrm{LS}})$,
with equality if and only if $\theta^{*}=\theta_{\textrm{LS}}$.
\end{thm}
Theorem \ref{thm:Interpretation}(i) shows that under misspecification
the function $x\mapsto x'\beta^{*}$ is an infeasible MVR approximation
of the true CMF penalized by the  mean ratio of the true variance
over its standard deviation approximation. An equivalent formulation
is
\begin{equation}
\min_{\theta\in\Theta}\;\frac{1}{2}E\left[\frac{1}{s(X'\gamma)}(\mu(X)-X'\beta)^{2}\right]+\frac{1}{2}E\left[\left\{ \frac{\sigma(X)^{2}}{s(X'\gamma)^{2}}+1\right\} s(X'\gamma)\right],\label{eq:PenWLS}
\end{equation}
the penalized WLS interpretation of the MVR problem (\ref{eq:theta*}).

The penalty term in (\ref{eq:PenWLS}) is a functional of a weighted
mean variance ratio of the true variance over its approximation. The
first-order conditions (\ref{eq:FOCmiss1})-(\ref{eq:FOCmiss2}) shed
additional light on how the weights are determined as well as on the
form of the penalty, by characterizing the optimality properties of
MVR approximations. Because $X$ includes an intercept, when both
functions $x\mapsto\mu(x)$ and $x\mapsto\sigma(x)^{2}$ are misspecified,
$\beta^{*}$ and $\gamma^{*}$ are chosen such that the sum of the
weighted mean squared error for the conditional mean and the mean
variance ratio error is zero, balancing the two approximation errors.
When the scale function is linear the two types of approximation error
are equalized. For the exponential specification, the two types of
approximation error weighted by $\exp(X'\gamma)$ are equalized. The
MVR solution is thus determined by minimizing the weighted mean squared
error for the conditional mean, while simultaneously setting the weighted
mean variance ratio as close as possible to one.

Theorem \ref{thm:Interpretation}(ii) formalizes the approximation
guarantee of MVR in terms of the MVR criterion. For the linear and
exponential scale function specifications, the improvement of the
MVR solution relative to the OLS solution in MVR loss further guarantees
that optimal weights are selected such that the weighted mean squared
MVR prediction error for $Y$ is not larger than the root mean squared
OLS prediction error.
\begin{cor}
\label{cor:Approx Bound}If the scale function $t\mapsto s(t)$ is
specified as $s(t)=t$ or $s(t)=\exp(t)$, then
\[
E\left[\frac{1}{s(X'\gamma^{*})}(Y-X'\beta^{*})^{2}\right]\leq E\left[(Y-X'\beta_{\textrm{LS}})^{2}\right]^{\frac{1}{2}},
\]
with equality if and only if $\theta^{*}=\theta_{\textrm{LS}}$.
\end{cor}
Compared to OLS, improvement in MVR loss is also related to a key
robustness property of MVR under a Kullback-Leibler measure of divergence.
Define the scaled Gaussian density function
\begin{equation}
f_{\theta}\left(Y,X\right):=\frac{1}{s(X'\gamma)}\phi\left(\frac{Y-X'\beta}{s(X'\gamma)}\right),\label{eq:Gausdens}
\end{equation}
where $\phi(z)=\left(2\pi\right)^{-1/2}\exp(-z^{2}/2)$. The OLS solution
maximizes a restricted version, $E[\log f_{\theta}\left(Y,X\right)]|_{\gamma_{-1}=0}$,
of the  expected log-likelihood $E[\log f_{\theta}\left(Y,X\right)]$
over $\Theta_{\textrm{LS}}$, where the components of $\gamma$ except
the first are set to zero. The corresponding expected log-likelihood
value $E[\log f_{\theta_{\textrm{LS}}}\left(Y,X\right)]$ is no greater
than the value of the expected log-likelihood at the MVR solution:
\begin{equation}
E\left[\log f_{\theta_{\textrm{LS}}}\left(Y,X\right)\right]\leq E\left[\log f_{\theta^{*}}\left(Y,X\right)\right].\label{eq:ineq1}
\end{equation}
The MVR solution $\theta^{*}$ formally corresponds to an improvement
of the expected log-likelihood value over OLS, and therefore corresponds
to a probability distribution that is KLIC closer to the true data
probability distribution $f_{Y\mid X}(Y\mid X)$.

Define the quantity $\epsilon:=E\left[\log\left(\sigma_{\textrm{LS}}^{2}/s(X'\gamma^{*})^{2}\right)\right]$,
which is positive as shown in Appendix \ref{subsec:Proof-of-TheoremKLIC}.
Our next result summarises the key implication of (\ref{eq:ineq1}).
\begin{thm}
\label{thm:MLequivalence}Suppose that $E[|\log f_{Y\mid X}(Y\mid X)|]<\infty$,
$E\left[e(Y,X,\theta^{*})^{2}\right]\leq1+\epsilon$, and the scale
function $t\mapsto s(t)$ is specified as $s(t)=t$ or $s(t)=\exp(t)$.
Then the probability distribution $f_{\theta^{*}}$ corresponding
to $\theta^{*}$ satisfies
\begin{equation}
E\left[\log\left(\frac{f_{Y\mid X}(Y\mid X)}{f_{\theta^{*}}(Y,X)}\right)\right]\leq E\left[\log\left(\frac{f_{Y\mid X}(Y\mid X)}{f_{\theta_{\textrm{LS}}}(Y,X)}\right)\right],\label{eq:-1}
\end{equation}
with equality if and only if $\theta^{*}=\theta_{\textrm{LS}}$. 
\end{thm}
For the linear scale specification, the MVR first-order conditions
include the constraint $E\left[e(Y,X,\theta^{*})^{2}-1\right]=0$
so that the bound on $E\left[e(Y,X,\theta^{*})^{2}\right]$ is satisfied
by construction. For the exponential specification, the corresponding
constraint is $E\left[\exp(X'\gamma^{*})\left\{ e(Y,X,\theta^{*})^{2}-1\right\} \right]=0$
which can result in a value for $E\left[e(Y,X,\theta^{*})^{2}\right]$
that differs from one under misspecification. The bound then characterizes
the deviations from unit variance that preserve the validity of (\ref{eq:-1}).\footnote{For the numerical simulations in Section \ref{subsec:Numerical-Simulations}
the typical sample mean of an estimate $e(Y,X,\hat{\theta})^{2}$
we observe is smaller than or equal to one, when the scale function
is specified as $s(t)=\exp(t)$. It is an open question whether the
bound $E\left[e(Y,X,\theta^{*})^{2}\right]\leq1+\epsilon$ can be
binding.}

The approximation guarantee (\ref{eq:-1}) is a general result that
holds under misspecification of the conditional mean and/or variance
functions. When the mean is misspecified, it formally establishes
that the MVR approximation for the mean corresponds to a better model
than the OLS location model according to the classical KLIC for model
selection (e.g., \citealp{Akaike:1973}; \citealp{Sawa:1978}). Similarly
to the classical argument motivating the use of maximum likelihood
(ML) under misspecification, Theorem \ref{thm:MLequivalence} thus
provides an information-theoretic justification for the use of MVR
and a formal characterization of the robustness to misspecification
of MVR relative to OLS.

\subsection{Interpretation Under Variance Misspecification}

If the CMF is linear, Theorem \ref{thm:Interpretation} has important
additional implications for the robustness and optimality properties
of MVR solutions. The $k$ orthogonality conditions (\ref{eq:FOCmiss2})
are then sufficient to determine the scale parameter $\gamma^{*}$
since condition (\ref{eq:FOCmiss1}) is uniquely satisfied by $\beta=\beta_{0}$.
Thus in the classical particular case of the linear conditional mean
model, the MVR solution for $\beta$ is fully robust to misspecification
of the scale function. Consequently, when the CMF is correctly specified
the OLS and MVR solutions for $\beta$ coincide. In the special case
of linear scale specification, $Xs_{1}(X'\gamma)$ reduces to $X$.
Because $X$ includes an intercept, the scale parameter $\gamma^{*}$
is then chosen such that the MVR conditional variance approximation
also satisfies the remarkable property of zero mean variance ratio
error. 
\begin{cor}
\label{cor:CLM}If Assumptions \ref{ass:ScaleSpec}-\ref{ass:FullRank}
hold and $\mu(X)=X'\beta_{0}$ a.s., then $\beta^{*}=\beta_{0}$ and
$\gamma^{*}$ is solely determined by  the $k$ orthogonality conditions
\begin{equation}
E\left[Xs_{1}(X'\gamma)\left\{ \frac{\sigma(X)^{2}}{s(X'\gamma){}^{2}}-1\right\} \right]=0.\label{eq:optimality}
\end{equation}
In particular, for the linear specification $s(t)=t$, the conditional
variance approximating function $x\mapsto(x'\gamma^{*}){}^{2}$ satisfies
the optimality property $E[\{\sigma(X)^{2}/(X'\gamma^{*}){}^{2}\}-1]=0.$
\end{cor}
When the CMF is correctly specified an optimal characterization of
$\beta_{0}$ that will lead to an efficient estimator can be formulated
by GLS. Define
\[
f_{\beta}^{\dagger}\left(Y,X\right):=\frac{1}{\sigma(X)}\phi\left(\frac{Y-X'\beta}{\sigma(X)}\right).
\]
In the population, GLS maximizes the expected log-likelihood $E[\log f_{\beta}^{\dagger}\left(Y,X\right)]$
with respect to $\beta$, with solution $\beta_{0}$. Then we further
have
\begin{equation}
E\left[\log f_{\theta^{*}}\left(Y,X\right)\right]\leq E\left[\log f_{\beta_{0}}^{\dagger}\left(Y,X\right)\right],\label{eq:ineq2}
\end{equation}
and inequalities (\ref{eq:ineq1}) and (\ref{eq:ineq2}) together
imply that, compared to OLS, the MVR solution $\theta^{*}$ formally
corresponds to a probability distribution that is KLIC closer to the
reference probability distribution $f_{\beta_{0}}^{\dagger}\left(Y,X\right)$
associated to the GLS model.
\begin{thm}
\label{thm:MLequivalence2}Suppose that $E[|\log f_{Y\mid X}(Y\mid X)|]<\infty$,
$E\left[e(Y,X,\theta^{*})^{2}\right]\leq1+\epsilon$, and the scale
function $t\mapsto s(t)$ is specified as $s(t)=t$ or $s(t)=\exp(t)$.
If Assumptions \ref{ass:ScaleSpec}-\ref{ass:FullRank} hold and $\mu(X)=X'\beta_{0}$
a.s., then $f_{\theta^{*}}$ also satisfies
\begin{equation}
E\left[\log\left(\frac{f_{\beta_{0}}^{\dagger}(Y,X)}{f_{\theta^{*}}(Y,X)}\right)\right]\leq E\left[\log\left(\frac{f_{\beta_{0}}^{\dagger}(Y,X)}{f_{\theta_{\textrm{LS}}}(Y,X)}\right)\right],\label{eq:-2}
\end{equation}
with equality if and only if $\theta^{*}=\theta_{\textrm{LS}}$. 
\end{thm}
When the mean is correctly specified, all of the likelihood improvement
comes from selecting a better approximation for the standard deviation
function than OLS. Relative to the efficient GLS model for the mean,
inequality (\ref{eq:-2}) formally establishes that the OLS location
model is rejected against the MVR location-scale model according to
a likelihood ratio criterion (e.g., \citealp{V:1989}; \citealp{SW:2017}).
If the true conditional variance is not constant then the improvement
in (\ref{eq:-2}) is strict and the MVR model is closer to the efficient
GLS model than the OLS location model.

In the presence of heteroskedasticity, MVR optimality and approximation
properties (\ref{eq:optimality}) and (\ref{eq:-2}) under correct
mean specification provide a theoretical justification for the largely
improved MVR-based inference relative to OLS-based inference in the
numerical simulations of Section \ref{sec:Numerical-Illustrations}
and the Supplementary Material. In view of its interpretation and
since it always admits a well-defined minimizer, the MVR criterion
thus offers a natural generalization of OLS for the estimation of
linear models.

\subsection{Connection with Gaussian Maximum Likelihood}

MVR provides one criterion for the simultaneous approximation of conditional
mean and variance functions. A related criterion is the KLIC of the
scaled Gaussian density $f_{\theta}\left(Y,X\right)$ defined in (\ref{eq:Gausdens})
from the true conditional density function $f_{Y\mid X}(Y\mid X)$,
which is minimized at a ML  pseudo-true value (\citealp{White:1982}).
Define for $\theta\in\Theta$,
\begin{equation}
\mathcal{L}\left(\theta\right):=-E\left[\log f_{\theta}\left(Y,X\right)\right]=\frac{1}{2}\log\left(2\pi\right)+E\left[\log s(X'\gamma)+\frac{1}{2}e\left(Y,X,\theta\right)^{2}\right],\label{eq:ML}
\end{equation}
with first-order conditions
\begin{equation}
E\left[\frac{X}{s(X'\gamma)}e\left(Y,X,\theta\right)\right]=0,\quad E\left[X\frac{s_{1}(X'\gamma)}{s(X'\gamma)}\left\{ e\left(Y,X,\theta\right)^{2}-1\right\} \right]=0.\label{eq:FOC-ML}
\end{equation}
In general the MVR solution $\theta^{*}$ need not satisfy equations
(\ref{eq:FOC-ML}), and therefore cannot be interpreted as a ML pseudo-true
value. Compared with the MVR criterion, an important limitation of
criterion (\ref{eq:ML}) is its lack of convexity. The second-order
derivative of $\mathcal{L}\left(\theta\right)$ with respect to the
first component $\gamma_{1}$ of $\gamma$, i.e., for fixed $\beta$,
$\gamma_{-1}$, is
\[
\frac{\partial^{\textrm{2}}\mathcal{L}\left(\theta\right)}{\partial\gamma_{1}^{2}}=E\left[\frac{1}{s(X'\gamma)^{2}}\left\{ 3e(Y,X,\theta)^{2}-1\right\} \right],
\]
which is strictly negative for all $\theta\in\Theta$ such that $e\left(Y,X,\theta\right)^{2}\leq1/3$
a.s. The non convexity of (\ref{eq:ML}) in\textit{ }$\gamma_{1}$
implies that $\mathcal{L}\left(\theta\right)$ is not jointly convex\footnote{\citet{Owen:2007} also noted the lack of joint convexity of the negative
Gaussian log-likelihood when the scale function is specified to a
constant, i.e., for the case $s(X'\gamma)=\sigma\in(0,\infty)$ in
(\ref{eq:ML}).}, and that a ML pseudo-true value might not exist; even if there exists
one, it need not be unique. In the latter case, some solutions may
only be local minima of (\ref{eq:ML}), not endowed with a KLIC-closest
interpretation and thus no longer guaranteed to improve over OLS and
MVR in a meaningful way.

These observations together with Theorems \ref{thm:MLequivalence}
and \ref{thm:MLequivalence2} clarify the relationship between ML,
OLS and MVR approximating properties. The ML pseudo-true value is
the parameter value associated with the distribution which is KLIC
closest to the true data generating process, but is not well-defined
due to the objective's lack of convexity. The OLS pseudo-true value
is well-defined, but it maximizes a restricted version of the Gaussian
expected log-likelihood resulting in a relatively lower likelihood.
The MVR loss function strikes a compromise by providing a well-defined
convex alternative to Gaussian ML, and relative to OLS by selecting
a pseudo-true value that corresponds to a distribution which is KLIC
closer to the true data generating process, and KLIC closer to the
efficient GLS model under correct mean specification.

\section{Estimation and Inference\label{sec:Estimation-and-Inference}}

We use the sample analog of the MVR population problem (\ref{eq:MVR})
for estimation of its solution $\theta^{*}$ in finite samples. We
establish existence, uniqueness and consistency of the MVR estimator.
We also derive its asymptotic distribution allowing for misspecification
of the shapes of the conditional mean and variance functions, and
discuss the robustness properties of its influence function. Finally,
we provide corresponding tools for robust inference and introduce
a one-step MVR-based test for heteroskedasticity.

We assume that we observe a sample of $n$ independent and identically
distributed realizations $\{(y_{i},x_{i})\}_{i=1}^{n}$ of the random
vector $(Y,X)$. We denote the $n\times k$ matrix of explanatory
variables values by $X_{n}$. We define $\Theta_{n}=\mathbb{R}^{k}\times\Theta_{\gamma,n}$,
with $\Theta_{\gamma,n}=\{\gamma\in\mathbb{R}^{k}:s(x_{i}'\gamma)>0,i=1,\ldots,n\}$,
the sample analog of the parameter space $\Theta$. For $\gamma\in\Theta_{\gamma,n}$,
we let $\Omega_{n}(\gamma)=\textrm{diag}(s(x_{i}'\gamma))$, an $n\times n$
diagonal matrix with diagonal elements $s(x_{1}'\gamma),\dots,s(x_{n}'\gamma)$.
We also define the MVR moment functions
\[
m_{1}(y_{i},x_{i},\theta):=x_{i}e(y_{i},x_{i},\theta),\quad m_{2}(y_{i},x_{i},\theta):=\frac{1}{2}x_{i}s_{1}(x_{i}'\gamma)\{e(y_{i},x_{i},\theta)^{2}-1\},
\]
and the corresponding vector $m(y_{i},x_{i},\theta):=(m_{1}(y_{i},x_{i},\theta),m_{2}(y_{i},x_{i},\theta))'$.

\subsection{The MVR Estimator}

The solution to the finite-sample analog of problem (\ref{eq:MVR})
is the MVR estimator
\begin{equation}
\hat{\theta}:=\arg\min_{\theta\in\Theta_{n}}\frac{1}{n}\sum_{i=1}^{n}\frac{1}{2}\left\{ \left(\frac{y_{i}-x_{i}'\beta}{s(x_{i}'\gamma)}\right)^{2}+1\right\} s(x_{i}'\gamma).\label{eq:MVR estimator}
\end{equation}
For $a=0$ in Assumption \ref{ass:ScaleSpec}, the sample objective
in (\ref{eq:MVR estimator}) is minimized subject to the $n$ inequality
constraints $s(x_{i}'\gamma)>0$, $i=1,\ldots,n$. For $a=-\infty$,
the parameter space simplifies to $\Theta_{n}=\mathbb{R}^{2\times k}$
and problem (\ref{eq:MVR estimator}) is unconstrained. In terms of
implementation, this constitutes an attractive feature of the exponential
scale specification.

We derive the asymptotic properties of $\hat{\theta}$ under the following
assumptions stated for a scale function in the class defined by Assumption
\ref{ass:ScaleSpec}.

\begin{assumption}(i) $\{(y_{i},x_{i})\}_{i=1}^{n}$ are identically
and independently distributed, and (ii) for all $\gamma\in\Theta_{\gamma,n}$,
the matrix $X_{n}'\Omega_{n}^{-1}(\gamma)X_{n}$ is finite and positive
definite.\label{ass:Sample}\end{assumption}

\begin{assumption}$E[Y^{6}]<\infty$, $E[||X||^{6}]<\infty$, and
for all $\gamma\in\Theta_{\gamma}$, $E[\left\Vert X\right\Vert ^{6}s_{2}(X'\gamma)^{6}]<\infty$.\label{ass:ANorm}\end{assumption}

Assumption \ref{ass:Sample}(i) can be replaced with the condition
that $\{(y_{i},x_{i})\}_{i=1}^{n}$ is stationary and ergodic \citep{Newey:McFadden:1994}.
Assumption \ref{ass:ANorm} is needed for asymptotic normality of
estimates of $\theta^{*}$. When the scale function $t\mapsto s(t)$
is specified to be linear, this assumption simplifies to the requirement
that $E[Y^{6}]$ and $E[\left\Vert X\right\Vert ^{6}]$ be finite.

Letting $e=e(Y,X,\theta^{*})$, the variance-covariance matrix of
the MVR estimator $\hat{\theta}$ is $G^{-1}SG^{-1}/n$, where
\[
G:=\left[\begin{array}{cc}
G_{11} & G_{12}\\
G_{21} & G_{22}
\end{array}\right]:=E\left[\begin{array}{cc}
\frac{XX'}{s(X'\gamma^{*})} & \frac{XX'}{s(X'\gamma^{*})}s_{1}(X'\gamma^{*})e\\
\frac{XX'}{s(X'\gamma^{*})}s_{1}(X'\gamma^{*})e & XX'\left\{ \frac{(s_{1}(X'\gamma^{*})e)^{2}}{s(X'\gamma^{*})}-\frac{1}{2}s_{2}(X'\gamma^{*})(e{}^{2}-1)\right\} 
\end{array}\right]
\]
and
\[
S:=\left[\begin{array}{cc}
S_{11} & S_{12}\\
S_{21} & S_{22}
\end{array}\right]:=E\left[\begin{array}{cc}
XX'e^{2} & \frac{1}{2}XX's_{1}(X'\gamma^{*})e(e^{2}-1)\\
\frac{1}{2}XX's_{1}(X'\gamma^{*})e(e^{2}-1) & \frac{1}{4}XX'\{s_{1}(X'\gamma^{*})(e^{2}-1)\}^{2}
\end{array}\right].
\]
The exact form of each component of matrices $G$ and $S$ depends
on the specification of the conditional mean and variance functions,
and simplifications of the variance-covariance matrix occur according
to the type of misspecification. Under mean misspecification, the
form of the variance-covariance matrix of the MVR estimator is not
affected by the specification of the conditional variance function.

\begin{sloppy}Define estimates of $G$ and $S$ by $\hat{G}:=n^{-1}\sum_{i=1}^{n}\partial m(y_{i},x_{i},\hat{\theta})/\partial\theta$
and $\hat{S}:=n^{-1}\sum_{i=1}^{n}m(y_{i},x_{i},\hat{\theta})m(y_{i},x_{i},\hat{\theta})'$,
respectively. The next theorem states the asymptotic properties of
the MVR estimator.\par\end{sloppy}
\begin{thm}
If Assumptions \ref{ass:ScaleSpec}-\ref{ass:ANorm} hold, then (i)
there exists $\hat{\theta}$ in $\Theta$ with probability approaching
one; (ii) $\hat{\theta}\rightarrow^{p}\theta^{*}$; and (iii)
\begin{equation}
n^{1/2}(\hat{\theta}-\theta^{*})\rightarrow_{d}\mathcal{N}(0,G^{-1}SG^{-1}).\label{eq:AsyDist_theta}
\end{equation}
If $\mu(X)=X'\beta^{*}$ a.s., then the following simplifications
occur
\begin{equation}
G_{12}=G_{21}=0_{k\times k},\quad S_{12}=S_{21}=\frac{1}{2}E[XX's_{1}(X'\gamma^{*})e^{3}].\label{eq:Var_miss}
\end{equation}
If $\mu(X)=X'\beta^{*}$ a.s. and $\sigma(X)^{2}=s(X'\gamma^{*})^{2}$
a.s., then the following additional simplifications occur
\begin{equation}
G_{22}=E\left[\frac{XX'}{s(X'\gamma^{*})}s_{1}(X'\gamma^{*})\right],\quad S_{11}=E[XX'],\quad S_{22}=\frac{1}{4}E[XX's_{1}(X'\gamma^{*})^{2}(e^{4}-1)].\label{eq:Correct_spec}
\end{equation}
Moreover, $\hat{G}^{-1}\hat{S}\hat{G}^{-1}\rightarrow^{p}G^{-1}SG^{-1}$.
\label{thm:ANorm}
\end{thm}
Theorem \ref{thm:ANorm} allows the construction of confidence intervals
and the implementation of hypothesis tests for $\theta$ under each
type of model specification using standard errors constructed from
the corresponding variance-covariance matrix. The various forms of
the variance-covariance matrix in Theorem \ref{thm:ANorm} provide
a basis for the construction of a range of specification tests, similarly
to the information matrix equality test in ML theory (\citealp{White:1982};
\citealp{CS:1991}). Inference using the general asymptotic variance
formula in (\ref{eq:AsyDist_theta}) will automatically be robust
to all forms of misspecification, and therefore to the presence of
heteroskedasticity of unknown form. 

For the linear homoskedastic model $Y=X'\beta_{0}+U$ with $E[U\mid X]=0$
and $\textrm{Var}(U\mid X)=\sigma_{0}^{2}$, the MVR and OLS variance-covariance
matrices coincide asymptotically, and MVR is efficient. Our numerical
simulations in Section \ref{sec:Numerical-Illustrations} and the
Supplementary Material illustrate that there is close to no finite-sample
loss in estimating linear homoskedastic models using MVR instead of
OLS. If both the conditional mean and variance functions are correctly
specified, then GLS with weights $1/\sigma^{2}(x)$ is an efficient
estimator for $\beta_{0}$. Letting $\check{Y}:=Y/s(X'\gamma^{*})$,
$\check{X}:=X/s(X'\gamma^{*})$ and $\check{s}(X'\gamma):=s(X'\gamma)/s(X'\gamma^{*})$,
define the weighted MVR objective
\[
Q^{\textrm{WMVR}}(\theta):=E\left[\frac{1}{2}\left\{ \left(\frac{\check{Y}-\check{X}'\beta}{\check{s}(X'\gamma)}\right)^{2}+1\right\} \check{s}(X'\gamma)\right]=E\left[\frac{1}{2}\left\{ e\left(Y,X,\theta\right)^{2}+1\right\} \frac{s(X'\gamma)}{s(X'\gamma^{*})}\right].
\]
If $\sigma^{2}(X)=s(X'\gamma_{0})^{2}$, then $\gamma^{*}=\gamma_{0}$
and $Q^{\textrm{WMVR}}(\theta)$ has first-order conditions for $\beta$
\[
\frac{\partial Q^{\textrm{WMVR}}(\theta)}{\partial\beta}=-E\left[\frac{X}{s(X'\gamma_{0})}e\left(Y,X,\theta\right)\right]=0,
\]
which are satisfied by $\theta=\theta_{0}$ and coincide with the
GLS (and ML) first-order conditions for $\beta$ at a solution. In
general, the functional form of the conditional variance function
is unknown, and the MVR and weighted MVR solutions will differ.

An important implication of Theorem \ref{thm:ANorm} is that the influence
function of the MVR estimator for $\beta$ is proportional to both
moment functions $m_{1}$ and $m_{2}$:
\[
IF_{\beta}(y,x,\theta)=-(G_{11}-G_{12}G_{22}^{-1}G_{21})^{-1}[m_{1}(y,x,\theta)-G_{12}G_{22}^{-1}m_{2}(y,x,\theta)].
\]
The quadratic term $m_{2}$ dominates and an influential observation
is defined as having $(y_{i}-x_{i}'\beta)^{2}$ large enough for $e(y_{i},x_{i},\theta)^{2}$
to be large. Observations that are influential for $\beta$ are influential
relative to the dispersion of $Y$, accounting for mean misspecification.

When the CMF is well-specified, the variance-covariance matrix takes
the form
\[
G^{-1}SG^{-1}=\left[\begin{array}{cc}
G_{11}^{-1}S_{11}G_{11}^{-1} & G_{11}^{-1}S_{12}G_{22}^{-1}\\
G_{22}^{-1}S_{21}G_{11}^{-1} & G_{22}^{-1}S_{22}G_{22}^{-1}
\end{array}\right].
\]
The influence function of $\beta$ is then proportional to $m_{1}$
only, and the influence function of $\gamma$ is proportional to $m_{2}$
only, since the off-diagonal blocks of $G$ are then $0_{k\times k}$.
For the mean parameter $\beta$, an observation $(y_{i},x_{i})$ with
large influence will be such that $y_{i}$ is large enough for the
standardized residual $e(y_{i},x_{i},\theta)$ to be large. Because
$\hat{\beta}$ and $\hat{\gamma}$ are determined simultaneously,
the influence of outliers on the mean parameter is limited by the
restriction that the sample second moment of $e(y_{i},x_{i},\theta)$
must remain close to one, and be exactly one if the scale is linear.
In sharp contrast with OLS, the scale parameter will simultaneously
compensate an increase in $Y$ dispersion so as to keep the variance
of $e(y_{i},x_{i},\theta)$ constant. The MVR influence function although
unbounded for a fixed value of $\gamma$ thus robustifies OLS through
the simultaneous reweighting of the residuals, downweighting regions
in the covariate space where the information on $Y$ is imprecise,
as measured by $s(x'\gamma)$, in the calculation of the regression
fit. 

In summary, the MVR estimator does not robustify OLS through the bounding
of the influence function (\citealp{Koenker:2005}), but  by incorporating
information  about the dispersion of $Y$ across the covariate space
in the definition of an influential outlier.
\begin{rem}
(Implementation.) Under our assumptions, the MVR objective is globally
convex in $\theta$, and therefore in $\beta$ for any $\gamma\in\Theta_{\gamma,n}$.
This implies that for any $\gamma\in\Theta_{\gamma,n}$ there exists
a unique corresponding minimizer $\hat{\beta}(\gamma)$. This observation
forms the basis of our implementation, and letting $y=(y_{1},\ldots,y_{n})'$,
we first obtain $\hat{\gamma}$ by solving
\begin{align*}
\min_{\gamma\in\mathbb{R}^{k}}\,\frac{1}{n}\sum_{i=1}^{n}\frac{1}{2}\left\{ \left(\frac{y_{i}-x_{i}'\hat{\beta}(\gamma)}{s(x_{i}'\gamma)}\right)^{2}+1\right\} s(x_{i}'\gamma), & \quad\hat{\beta}(\gamma):=[X_{n}'\varOmega_{n}^{-1}(\gamma)X_{n}]^{-1}X_{n}'\varOmega_{n}^{-1}(\gamma)y,\\
\textrm{s.t.}\quad s(x_{i}'\gamma)>0,\quad i=1,\ldots,n, & \quad\textrm{if }s(t)\leq0\textrm{ for some }t\in\mathbb{R}.
\end{align*}
Concentrating out $\beta$ for each $\gamma$ provides a convenient
implementation of the MVR estimator $\hat{\gamma}$, with the final
estimate for $\beta$ defined as $\hat{\beta}:=\hat{\beta}(\hat{\gamma})$.\qed
\end{rem}

\subsection{Inference}

Given the MVR estimator $\hat{\theta}=(\hat{\beta},\hat{\gamma})'$,
inference is performed based on the estimated asymptotic variance-covariance
matrix $\hat{V}:=\hat{G}^{-1}\hat{S}\hat{G}^{-1}$, which can be partitioned
into 4 blocks
\[
\hat{V}=\left[\begin{array}{cc}
\hat{V}_{11} & \hat{V}_{12}\\
\hat{V}_{21} & \hat{V}_{22}
\end{array}\right].
\]
The specific form of $\hat{V}$ depends on the specification assumptions
made on the conditional mean and variance functions. For $\hat{\beta}_{j}$
and $\hat{\gamma}_{j}$ the $j$th components of $\hat{\beta}$ and
$\hat{\gamma}$, respectively, MVR standard errors are obtained as
\[
\textrm{s.e.}(\hat{\beta}_{j}):=\left(\frac{1}{n}\left[\hat{V}_{11}\right]_{j,j}\right)^{\frac{1}{2}},\quad\textrm{s.e.}(\hat{\gamma}_{j}):=\left(\frac{1}{n}\left[\hat{V}_{22}\right]_{j,j}\right)^{\frac{1}{2}},
\]
with resulting two-sided confidence intervals with nominal level $1-\alpha$,
\[
\hat{\beta}_{j}\pm\Phi^{-1}(1-\alpha/2)\times\textrm{s.e.}(\hat{\beta}_{j}),\quad\hat{\gamma}_{j}\pm\Phi^{-1}(1-\alpha/2)\times\textrm{s.e.}(\hat{\gamma}_{j}),
\]
where $\Phi^{-1}(1-\alpha/2)$ denotes the $1-\alpha/2$ quantile
of the Gaussian distribution. A significance test of the null $\beta_{j}=0$
and $\gamma_{j}=0$ can then be performed using the test statistics
$\hat{\beta}_{j}/\textrm{s.e.}(\hat{\beta}_{j})$ and $\hat{\gamma}_{j}/\textrm{s.e.}(\hat{\gamma}_{j})$.

Simultaneous significance testing or hypothesis tests on linear combination
of multiple parameters can be implemented by a Wald test. For $h\leq2\times k$,
letting \textbf{$R$ }be an $h\times(2\times k)$ matrix of constants
of full rank $h$ and $r$ be an $h\times1$ vector of constants,
define
\[
H_{0}:R\theta^{*}-r=0,\quad H_{1}:R\theta^{*}-r\neq0,
\]
the null and alternative hypotheses for a two-sided tests of linear
restrictions on the location-scale model $Y=X'\beta^{*}+s(X'\gamma^{*})e$.
It follows from asymptotic normality of $\hat{\theta}$ in (\ref{eq:AsyDist_theta})
that the corresponding MVR Wald statistic $W_{\textrm{MVR}}$ satisfies
\[
W_{\textrm{MVR}}:=(R\hat{\theta}-r)'[R(\hat{V}/n)R']^{-1}(R\hat{\theta}-r)\sim\chi_{(h)}^{2},
\]
under the null $H_{0}$.

The Wald statistic $W_{\textrm{MVR}}$ can be specialized to formulate
a one-step robust MVR-based test for heteroskedasticity. Letting 
\[
h=k-1,\quad R=\left[\begin{array}{cc}
0_{k-1,k+1} & I_{k-1}\end{array}\right],\quad r=0_{k-1},
\]
the statistic $W_{\textrm{MVR}}$ gives a robust test of the null
hypothesis $H_{0}:\gamma_{2}^{*}=\ldots=\gamma_{k}^{*}=0$.
\begin{rem}
When the CMF is linear, robust MVR inference on $\hat{\beta}$ uses
the closed-form variance formula
\[
\widehat{\textrm{Var}}(\hat{\beta})=n^{-1}(X_{n}'\Omega_{n}^{-1}(\hat{\gamma})X_{n})^{-1}(X_{n}'\hat{\varPsi}_{e}X_{n})(X_{n}'\Omega_{n}^{-1}(\hat{\gamma})X_{n})^{-1},
\]
where $\hat{\varPsi}_{e}=\textrm{diag}(\hat{e}_{i}^{2})$.\qed
\end{rem}

\section{Numerical Illustrations\label{sec:Numerical-Illustrations}}

All computational procedures can be implemented in the software R
(\citealp{R:2017}) using open source software packages for nonlinear
optimization such as Nlopt, and its R interface Nloptr (\citet{YBE:2018}).

\subsection{Empirical Application: Reversal of Fortune}

We apply our methods to the study of the effect of European colonialism
on today's relative wealth of former colonies, as in \citet{AJR:2002}.
They show that former colonies that were relatively rich in 1500 are
now relatively poor, and provide ample empirical evidence of this
reversal of fortune. In particular, they study the relationship between
urbanization in 1500 and GDP per capita in 1995 (PPP basis), using
OLS regression analysis. The sample size ranges from 17 to 41 former
colonies, allowing the illustration of MVR properties in small samples.

We take the outcome $Y$ to be log GDP per capita in 1995 and in the
baseline specification $X$ includes an intercept and a measure of
urbanization in 1500, a proxy for economic development. We implement
MVR with both linear ($\ell$-MVR) and exponential ($e$-MVR) scale
functions, and report estimated standard errors robust to mean misspecification
according to (\ref{eq:AsyDist_theta}). We also report OLS estimates,
with heteroskedasticity-robust standard errors. In the Supplementary
Material we also provide results including standard errors with finite-sample
adjustments suggested by \citet{MKW:1985}, and we also report MVR
standard errors calculated under correct mean misspecification. Our
findings are robust to these modifications.

\begin{sidewaystable}
\vspace*{+16cm} %
\begin{tabular}{lcccccccccccc}
\toprule 
 &  & \multicolumn{11}{c}{Dependent variable is log GDP per capita (PPP) in 1995}\tabularnewline
 &  &  &  &  &  &  &  &  &  &  &  & \tabularnewline
\cmidrule{3-13} 
 &  & OLS  & $\ell$-MVR  & \textit{e}-MVR  &  & OLS  & $\ell$-MVR  & \textit{e}-MVR  &  & OLS  & $\ell$-MVR  & \textit{e}-MVR \tabularnewline
\cmidrule{3-5} \cmidrule{7-9} \cmidrule{11-13} 
 &  & \multicolumn{3}{c}{(1) Base sample} &  & \multicolumn{3}{c}{(2) Without North Africa} &  & \multicolumn{3}{c}{(3) Without the Americas}\tabularnewline
 &  & \multicolumn{3}{c}{($n=41$)} &  & \multicolumn{3}{c}{($n=37$)} &  & \multicolumn{3}{c}{($n=17$)}\tabularnewline
Urbanization in 1500  &  & -0.078  & -0.067  & -0.069  &  & -0.101  & -0.099  & -0.099  &  & -0.115  & -0.064  & -0.077 \tabularnewline
 &  & (0.023) & (0.028) & (0.026) &  & (0.032) & (0.034) & (0.034) &  & (0.043) & (0.127) & (0.113)\tabularnewline
 &  &  &  &  &  &  &  &  &  &  &  & \tabularnewline
\cmidrule{3-5} \cmidrule{7-9} \cmidrule{11-13} 
 &  & \multicolumn{3}{c}{(4) Just the Americas} &  & \multicolumn{3}{c}{(5) With the continent} &  & \multicolumn{3}{c}{(6) Without neo-Europes}\tabularnewline
 &  & \multicolumn{3}{c}{($n=24$)} &  & \multicolumn{3}{c}{dummies ($n=41$)} &  & \multicolumn{3}{c}{($n=37$)}\tabularnewline
Urbanization in 1500  &  & -0.053  & -0.045  & -0.044  &  & -0.082  & -0.063  & -0.060  &  & -0.046  & -0.036  & -0.038 \tabularnewline
 &  & (0.029) & (0.032) & (0.032) &  & (0.031) & (0.029) & (0.030) &  & (0.021) & (0.023) & (0.023)\tabularnewline
 &  &  &  &  &  &  &  &  &  &  &  & \tabularnewline
\cmidrule{3-5} \cmidrule{7-9} \cmidrule{11-13} 
\textbf{ } &  & \multicolumn{3}{c}{(7) Controlling for Latitude} &  & \multicolumn{3}{c}{(8) Controlling for colonial} &  & \multicolumn{3}{c}{(9) Controlling for religion}\tabularnewline
 &  & \multicolumn{3}{c}{($n=41$)} &  & \multicolumn{3}{c}{origin ($n=41$)} &  & \multicolumn{3}{c}{($n=41$)}\tabularnewline
Urbanization in 1500  &  & -0.072  & -0.069  & -0.070  &  & -0.071  & -0.063  & -0.062  &  & -0.060  & -0.042  & -0.040 \tabularnewline
 &  & (0.020) & (0.022) & (0.021) &  & (0.025) & (0.026) & (0.027) &  & (0.027) & (0.029) & (0.029)\tabularnewline
 &  &  &  &  &  &  &  &  &  &  &  & \tabularnewline
\bottomrule
\end{tabular}\medskip{}

\caption{Reversal of fortune. Asymptotic heteroskedasticity-robust OLS standard
errors and MVR standard errors are in parenthesis.\label{tab:Reversal-of-fortuneCoefs}}
\end{sidewaystable}
Table \ref{tab:Reversal-of-fortuneCoefs} reports our results for
urbanization in the baseline specification across 5 different sets
of countries, and for 4 additional specifications\footnote{We exclude two specifications of Table III in \citet{AJR:2002} for
which not all types of OLS and MVR standard errors are well-defined.} including continent dummies, and controlling for latitude, colonial
origin and religion\footnote{See \citet{AJR:2002} for a detailed description of the data.}.
A striking feature of the results is the robustness to scale specification
of MVR point estimates and standard errors. They are nearly identical
across all specifications, except for Panel (3). Compared to OLS,
MVR point estimates are all smaller in magnitude, suggesting a negative
bias of OLS away from zero while standard errors are of similar magnitude,
making it more likely to find a significant relationship with OLS
estimates in this empirical application. The urbanization coefficient
loses significance with MVR in 4 specifications, mainly as a result
of the change in point estimates. 

Specifically, we find that MVR provides supporting evidence of a significant
statistical relationship between urbanization in 1500 and GDP per
capita in 1995 in the whole sample, but also dropping North Africa,
including continent dummies, and controlling for latitude and for
colonial origin. However, the relationship between urbanization in
1500 and GDP per capita in 1995 is not statistically significant in
the four remaining specifications. When the Americas are dropped (Panel
(3)), when only former colonies from the Americas are considered (Panel
(4)), and when controlling for religion (Panel (9)), the urbanization
coefficient is no longer significant with MVR. These results are robust
to implementing finite-sample adjustments. Specification (6) drops
observations for neo-Europes (United States, Canada, New Zealand,
and Australia), and only the $e$-MVR estimate is significant at the
10 percent level when no finite-sample adjustments are implemented,
and loses significance otherwise.

MVR results provide renewed empirical support for a subset of the
specifications, but overall show that the mean relationship in this
empirical application is weaker and less robust than first suggested
by the OLS-based analysis\footnote{We also implemented WLS and found that the magnitude of most WLS coefficients
is smaller than MVR point estimates. In addition to specifications
(3), (4), (6) and (9), specification (5) is also found to be not statistically
significant. We report the results in Section 3 of the Supplementary
Material.}. These findings illustrate that MVR can substantially alter the conclusions
obtained using OLS in practice.

\subsection{Numerical Simulations\label{subsec:Numerical-Simulations}}

We investigate the properties of MVR in small samples and compare
its performance with OLS and WLS by implementing numerical simulations
based on the experimental setup in \citet{MacKinnon:2013}. In the
Supplementary Material, we provide additional results for models featuring
a nonlinear CMF and report simulation results from an experiment calibrated
to a second empirical example. We find that using MVR approximations
does not result in a loss in the quality of approximation of nonlinear
CMFs compared to OLS, and MVR estimation and inference finite-sample
properties compare favorably to both OLS and WLS. 

\subsubsection{Design of Experiments}

The data generating process is
\begin{align*}
Y & =\beta_{0}+X_{1}\beta_{1}+X_{2}\beta_{2}+X_{3}\beta_{3}+X_{4}\beta_{4}+\sigma\varepsilon,\quad\varepsilon\sim\mathcal{N}(0,1)\\
\sigma & =z(\alpha)\left(\gamma_{0}+X_{1}\gamma_{1}+X_{2}\gamma_{2}+X_{3}\gamma_{3}+X_{4}\gamma_{4}\right)^{\alpha},\quad\alpha\in\{0,0.5,1,1.5,2\},
\end{align*}
where all regressors are drawn from the standard log-normal distribution,
and $z(\alpha)$ is chosen such that the expected variance of $\sigma\varepsilon$
is equal to 1. The log-normal regressors ensure that many samples
will include high-leverage points with a few observations taking extreme
values. This feature of the design distorts the distribution of test
statistics based on heteroskedasticity-robust estimators of OLS standard
errors. The parameter coefficient values are set to $\beta_{j}=\gamma_{j}=1$
for $j=0,\ldots,4$.\footnote{This departs slightly from the original \citet{MacKinnon:2013} design
where $\beta_{4}=\gamma_{4}=0$. We are grateful to James MacKinnon
for suggesting this modification that preserves heteroskedasticity
in $X_{4}$.} The index $\alpha$ measures the degree of heteroskedasticity in
the model, with $\alpha=0$ corresponding to homoskedasticity, and
$\alpha=2$ corresponding to high heteroskedasticity. The numerical
simulations are implemented for sample sizes $n=20,40,80,160,320,640$
and $1280$. 

For each $\alpha$ and sample size, we generate 10000 samples, and
implement OLS, WLS and MVR. We implement MVR with both linear ($\ell$-MVR)
and exponential ($e$-MVR) scale functions. For WLS we follow the
implementations proposed by Romano and Wolf (\citeyear{RW:2017},
cf. equation (3.4) and (3.5), p. 4). Denote the OLS estimator by $\hat{\beta}_{\textrm{LS}}$
and let $\tilde{x}_{i}=(x_{1i},x_{2i},x_{3i},x_{4i})'$. We form the
OLS residuals $\hat{u}_{i}:=y_{i}-x_{i}'\hat{\beta}_{\textrm{LS}}$,
$i=1,\ldots,n$, and perform the OLS regressions
\[
\log(\max(\delta^{2},\hat{u}_{i}^{2}))=\nu+\pi'|\tilde{x}_{i}|+\eta_{i},
\]
for WLS with linear scale ($\ell$-WLS)\footnote{This regression is performed imposing the $n$ constraints $\nu+\pi|\tilde{x}_{i}|\geq\delta$,
$i=1,\ldots,n$, using the $\mathtt{lsei}$ R package (\citealp{WLH:2017}). }, and
\[
\log(\max(\delta^{2},\hat{u}_{i}^{2}))=\nu+\pi'\log(|\tilde{x}_{i}|)+\eta_{i},
\]
for WLS with exponential scale ($e$-WLS), where $\delta=0.1$ as
in the implementation of \citet{RW:2017}, and with estimates $(\hat{\nu},\hat{\pi})$.
The WLS weights are formed as $\hat{w}_{i}^{\ell}:=\hat{\nu}+\hat{\pi}'|\tilde{x}_{i}|$
and $\hat{w}_{i}^{e}:=\exp(\hat{\nu}+\hat{\pi}'\log(|\tilde{x}_{i}|))$,
and the WLS estimators are
\[
\hat{\beta}_{\textrm{WLS}}^{m}:=[X'_{n}(W_{n}^{m})^{-1}X{}_{n}]^{-1}X'_{n}(W_{n}^{m})^{-1}y,\quad W_{n}^{m}:=\textrm{diag}(\hat{w}_{i}^{m}),\quad m=\ell,e,
\]
where $y=(y_{1},\ldots,y_{n})'$, $X_{n}$ is the $n\times5$ matrix
of explanatory variables values, and $\textrm{diag}(\hat{w}_{i}^{\ell})$
and $\textrm{diag}(\hat{w}_{i}^{e})$ denote the $n\times n$ diagonal
matrices with diagonal elements $\hat{w}_{1}^{\ell},\dots,\hat{w}_{n}^{\ell}$
and $\hat{w}_{1}^{e},\dots,\hat{w}_{n}^{e}$, respectively.

In all experiments the results for $\beta_{1}$ to $\beta_{4}$ are
similar and we thus only report the results for $\beta_{4}$ for brevity.
Also, throughout the relative performance of MVR and WLS is assessed
by comparing $\ell$-MVR to $\ell$-WLS, and $e$-MVR to $e$-WLS.

\subsubsection{Estimation results}

Tables \ref{tab:ARMSE_OLS} and \ref{tab:ARMSE_WLS} report the ratio
of MVR root mean squared errors (RMSE) across simulations over the
OLS and WLS RMSEs for the coefficient parameter $\beta_{4}$, each
sample size and value of heteroskedasticity parameter $\alpha$, in
percentage terms. Denoting an estimator $\tilde{\beta}_{4}^{(s)}$
of $\beta_{4}$ for the $s$th simulation, the RMSE is computed as
$\{\frac{1}{S}\sum_{s=1}^{S}(\tilde{\beta}_{4}^{(s)}-\beta_{4})^{2}\}^{1/2}$,
for $S=10000$. 

Table \ref{tab:ARMSE_OLS} shows that the performance of both MVR
estimators relative to OLS improves as $n$ and $\alpha$ increase.
As expected, for the homoskedastic case $\alpha=0$ the performance
of MVR and OLS estimators is very similar, and the ratios converge
to 100 from above, reflecting the efficiency of the OLS estimator
in that case. The performance of MVR then becomes markedly superior
as $n$ and $\alpha$ increase, with ratios that reach $31.7$ for
$\ell$-MVR and $22.9$ for \textit{e}-MVR. The estimator $\ell$-MVR
dominates \textit{e}-MVR slightly for $n=20$. The performance of
the estimator \textit{e}-MVR then becomes superior as the degree of
heteroskedasticity and sample size increase, showing higher robustness
of the exponential scale specification in more extreme designs in
these simulations.

In Table \ref{tab:ARMSE_WLS}, we find that the relative performance
of both MVR estimators relative to WLS also improves as $n$ increases
and as $\alpha$ increases from 0.5 to 2. For the homoskedastic case
$\alpha=0$, an interesting feature of the simulation results is that
the relative performance of MVR and WLS estimators now converges to
100 from below. This reflects the fact that for homoskedastic designs
MVR weights are better able to mitigate the cost of reweighting in
small samples compared to WLS weights. For other designs with $\alpha>0$,
the relative performance of both MVR estimators dominates the performance
of WLS with ratios that reach $76.2$ for $\ell$-MVR and $43.5$
for \textit{e}-MVR. For $\alpha=1$, WLS with linear scale is efficient
and dominates slightly $\ell$-MVR as $n$ increases. Compared to
OLS and the results of Table \ref{tab:ARMSE_OLS}, these results show
that in this experiment WLS also improves over OLS, that $\ell$-MVR
improves over WLS as $n$ increases and $\alpha$ deviates from $1$,
with little loss for $\alpha\leq1$, and $e$-MVR yields substantial
additional gains over WLS as both $n$ and $\alpha$ increase.

\begin{table}[t]
\centering %
\begin{tabular}{ccccccccccccc}
\toprule 
 &  & \multicolumn{5}{c}{$\ell$-MVR} &  & \multicolumn{5}{c}{\textit{e}-MVR}\tabularnewline
\midrule 
$\alpha$ &  & $0$ & $0.5$ & $1$ & $1.5$ & $2$ &  & $0$ & $0.5$ & $1$ & $1.5$ & $2$\tabularnewline
\midrule
$n=20$ &  & 105.5  & 104.1  & 98.9  & 91.9  & 84.9  &  & 106.7  & 105.2  & 99.7  & 92.4  & 84.4 \tabularnewline
$n=40$ &  & 103.7  & 100.1  & 90.9  & 79.6  & 69.2  &  & 103.4  & 100.0  & 91.1  & 79.5  & 67.2 \tabularnewline
$n=80$ &  & 102.3  & 96.7  & 84.0  & 69.7  & 58.5  &  & 102.1  & 96.9  & 84.3  & 69.1  & 54.6 \tabularnewline
$n=160$ &  & 101.7  & 94.0  & 77.4  & 60.5  & 49.5  &  & 101.5  & 94.2  & 77.4  & 58.9  & 43.2 \tabularnewline
$n=320$ &  & 101.1  & 91.3  & 71.5  & 53.0  & 42.6  &  & 100.9  & 91.4  & 70.7  & 50.2  & 34.6 \tabularnewline
$n=640$ &  & 100.7  & 89.5  & 66.8  & 47.1  & 37.3  &  & 100.6  & 89.6  & 66.1  & 44.1  & 28.7 \tabularnewline
$n=1280$ &  & 100.5  & 86.6  & 61.1  & 40.8  & 31.7  &  & 100.4  & 86.6  & 60.1  & 37.3  & 22.9 \tabularnewline
 &  &  &  &  &  &  &  &  &  &  &  & \tabularnewline
\bottomrule
\end{tabular}
\begin{centering}
\medskip{}
\par\end{centering}
\caption{Ratio ($\times100$) of MVR RMSE for $\beta_{4}$ over corresponding
OLS counterpart.\label{tab:ARMSE_OLS}}
\bigskip{}

\centering %
\begin{tabular}{ccccccccccccc}
\toprule 
 &  & \multicolumn{5}{c}{$\ell$-MVR} &  & \multicolumn{5}{c}{\textit{e}-MVR}\tabularnewline
\midrule 
$\alpha$ &  & $0$ & $0.5$ & $1$ & $1.5$ & $2$ &  & $0$ & $0.5$ & $1$ & $1.5$ & $2$\tabularnewline
\midrule
$n=20$ &  & 98.7  & 99.0  & 99.2  & 98.7  & 97.1  &  & 97.3  & 98.4  & 99.4  & 99.6  & 97.8 \tabularnewline
$n=40$ &  & 97.1  & 97.8  & 98.6  & 97.4  & 93.1  &  & 94.9  & 96.2  & 97.4  & 96.5  & 90.2 \tabularnewline
$n=80$ &  & 96.5  & 98.1  & 99.7  & 96.2  & 89.6  &  & 97.2  & 98.0  & 97.8  & 93.6  & 82.5 \tabularnewline
$n=160$ &  & 96.7  & 99.2  & 101.7  & 94.4  & 86.7  &  & 97.9  & 96.9  & 93.1  & 85.0  & 70.1 \tabularnewline
$n=320$ &  & 96.9  & 100.3  & 102.1  & 92.4  & 85.2  &  & 99.0  & 96.0  & 89.0  & 77.8  & 59.8 \tabularnewline
$n=640$ &  & 97.4  & 101.3  & 103.4  & 91.0  & 80.4  &  & 99.8  & 95.8  & 85.7  & 71.0  & 50.8 \tabularnewline
$n=1280$ &  & 98.1  & 102.3  & 104.1  & 88.8  & 76.2  &  & 99.8  & 94.2  & 81.6  & 64.7  & 43.5 \tabularnewline
 &  &  &  &  &  &  &  &  &  &  &  & \tabularnewline
\bottomrule
\end{tabular}
\begin{centering}
\medskip{}
\par\end{centering}
\caption{Ratio ($\times100$) of MVR RMSE for $\beta_{4}$ over corresponding
WLS counterpart.\label{tab:ARMSE_WLS}}
\end{table}

\subsubsection{Inference}

\begin{sloppy}In order to study the finite-sample performance of
MVR inference relative to heteroskedasticity-robust OLS and WLS inference,
we first compare the rejection probabilities of asymptotic $t$ tests
of the null hypothesis $\beta_{4}=1$ based on the standard normal
distribution\footnote{We also performed simulations with and tested for $\beta_{4}=0$,
and calculated rejection probabilities using a $t_{n-k}$ distribution.
The relative performance of the methods remains similar.}. We then compare the lengths of the confidence intervals constructed
for the coefficient parameter $\beta_{4}$. MVR standard errors are
calculated under mean misspecification according to (\ref{eq:AsyDist_theta}).
OLS and WLS standard errors used in the construction of confidence
intervals and tests statistics are the asymptotic heteroskedasticity-robust
standard errors. For completeness, in the Supplementary Material we
also compare rejection probabilities and confidence intervals based
on standard errors with finite-sample adjustments suggested by \citet{MKW:1985},
and we also replicate all experiments using MVR standard errors calculated
under correct mean misspecification with the simplifications in (\ref{eq:Var_miss}).\par\end{sloppy}

Figure \ref{fig:RejProbMVR1_HC0} displays rejection probability curves
of asymptotic $t$ tests of the null hypothesis $\beta_{4}=1$ for
each sample size and value of the heteroskedasticity parameter $\alpha$.
The nominal size of the tests is set to $5\%$. Figures \ref{fig:RejProbMVR1_HC0}(a)-(d)
show that MVR addresses overrejection of the OLS- and WLS-based tests
in the presence of heteroskedasticity ($\alpha>0$). A striking feature
of MVR rejection probability curves is that they flatten very quickly
across $\alpha$ as $n$ increases, a feature somewhat more pronounced
for $\ell$-MVR curves. This is in sharp contrast with OLS and $e$-WLS
rejection probability curves that are increasing with the degree of
heteroskedasticity $\alpha$. For $\ell$-WLS the rejection curves
are distorted around $\alpha=1$, for which it is efficient, and overall
the rejection probabilities are much larger than for $\ell$-MVR.
The curves for $\ell$-MVR and $\ell$-WLS coincide only for the case
where $\ell$-MVR is efficient ($\alpha=1$) at moderate sample size
and above ($n\geq320$). The $\ell$-MVR rejection probability curve
for $n=20$ (black curve) is not placed above the other curves although
it is above the nominal level for all values of $\alpha$. This feature
disappears when finite-sample corrections are implemented (Figures
2.1-2.3 in the Supplementary Material). 

\begin{figure}[!tph]
\subfloat[$\ell$-MVR vs OLS.]{\includegraphics[width=7.8cm,height=7.8cm]{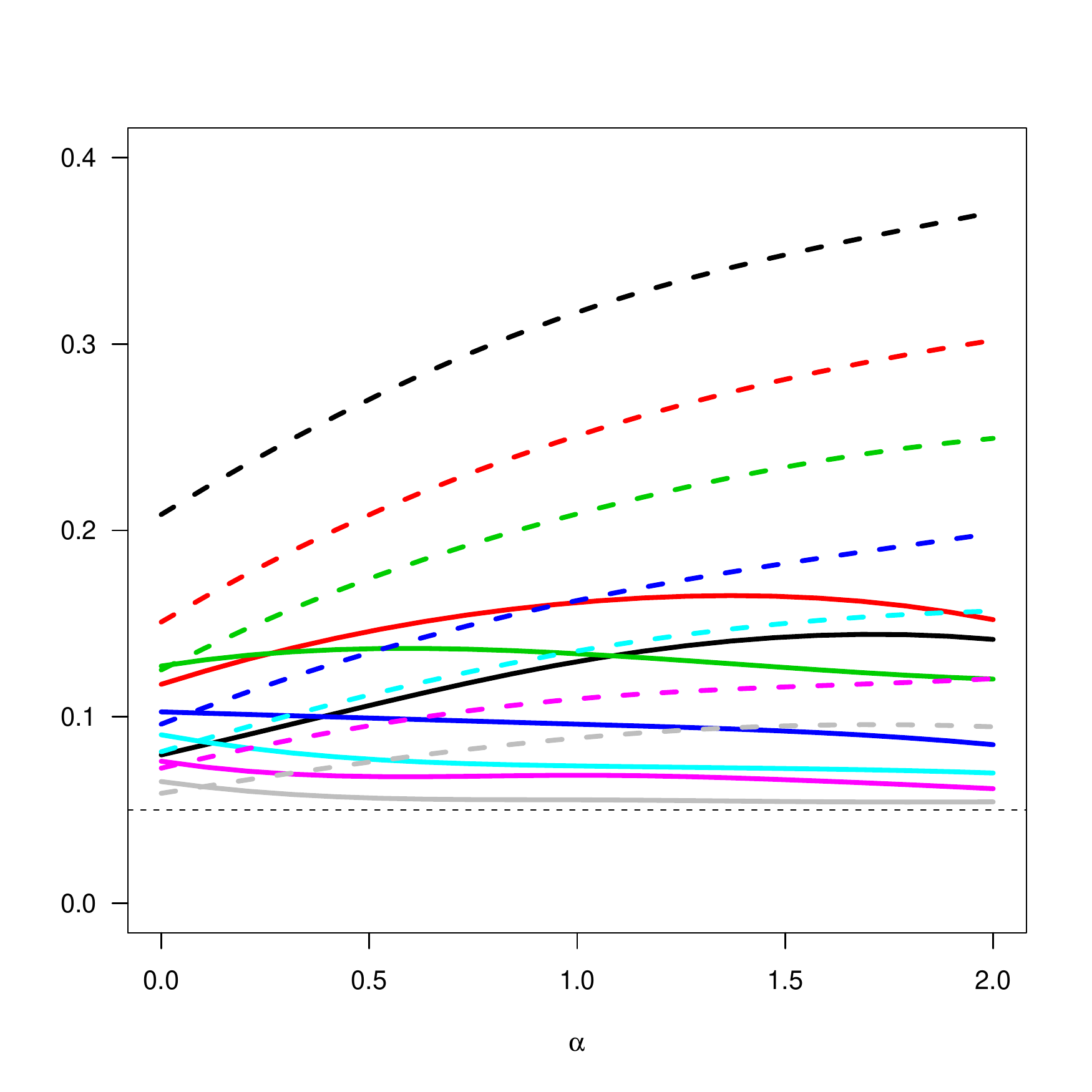}

}\subfloat[\textit{e}-MVR vs OLS.]{\includegraphics[width=7.8cm,height=7.8cm]{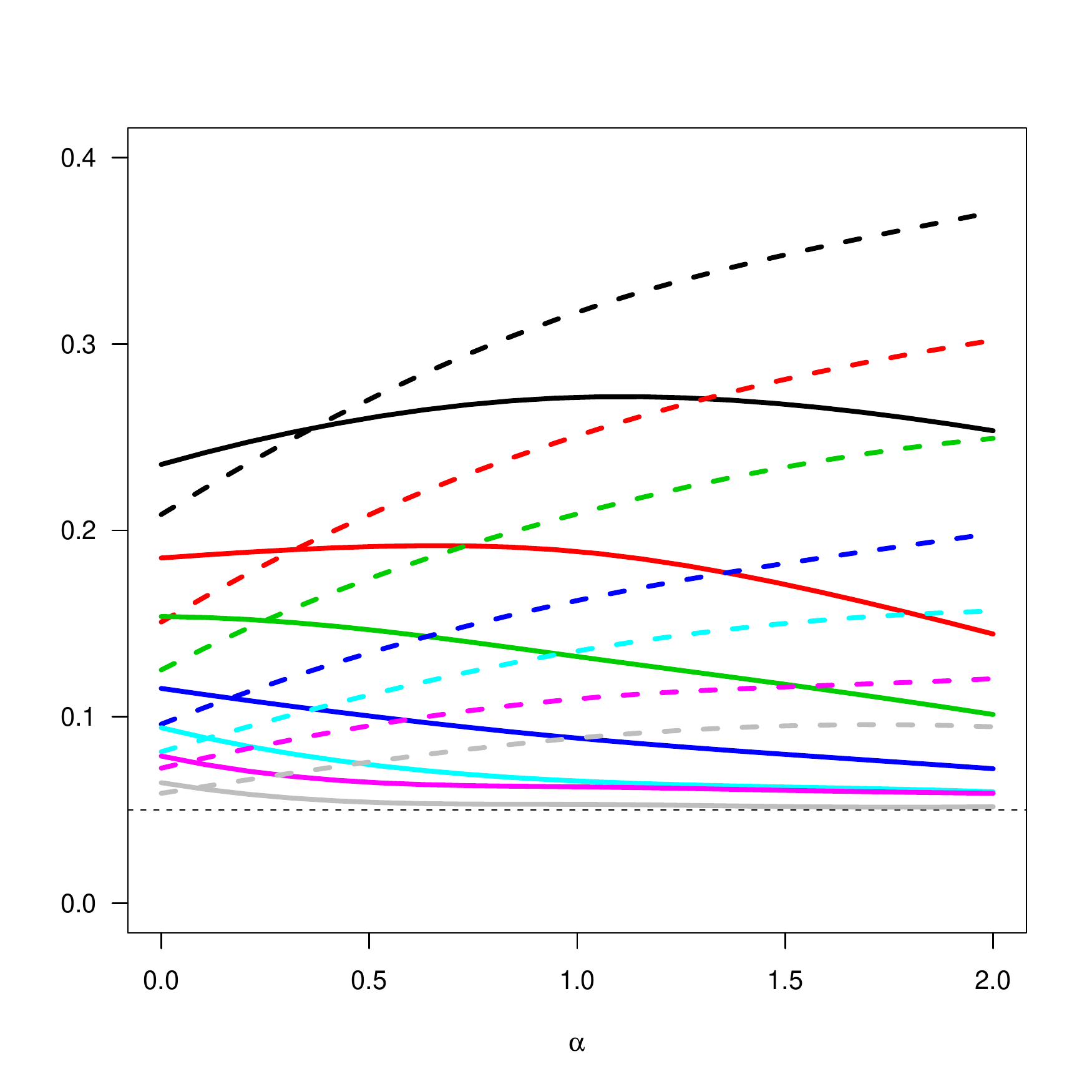}

}

\subfloat[$\ell$-MVR vs $\ell$-WLS.]{\includegraphics[width=7.8cm,height=7.8cm]{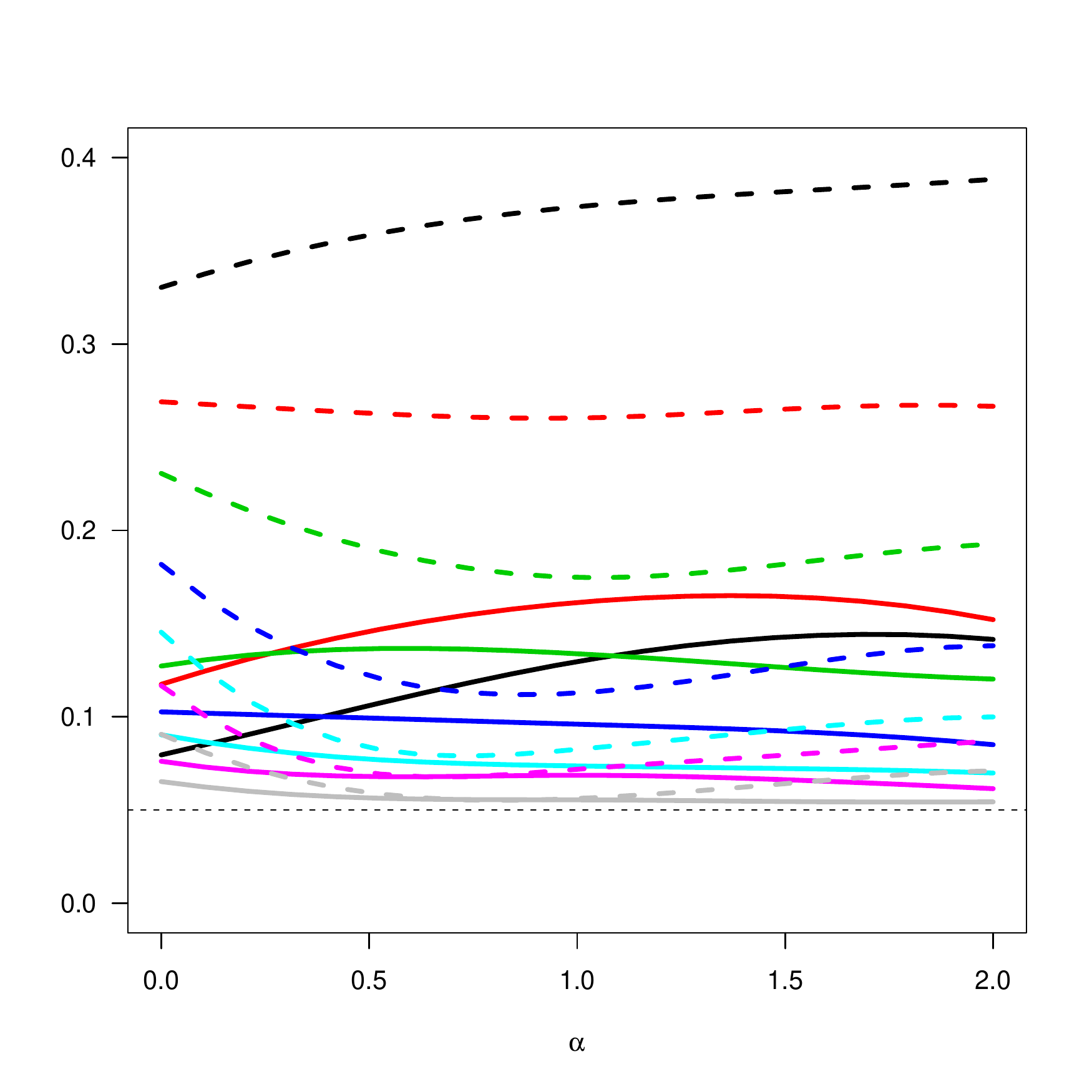}

}\subfloat[\textit{e}-MVR vs \textit{e}-WLS.]{\includegraphics[width=7.8cm,height=7.8cm]{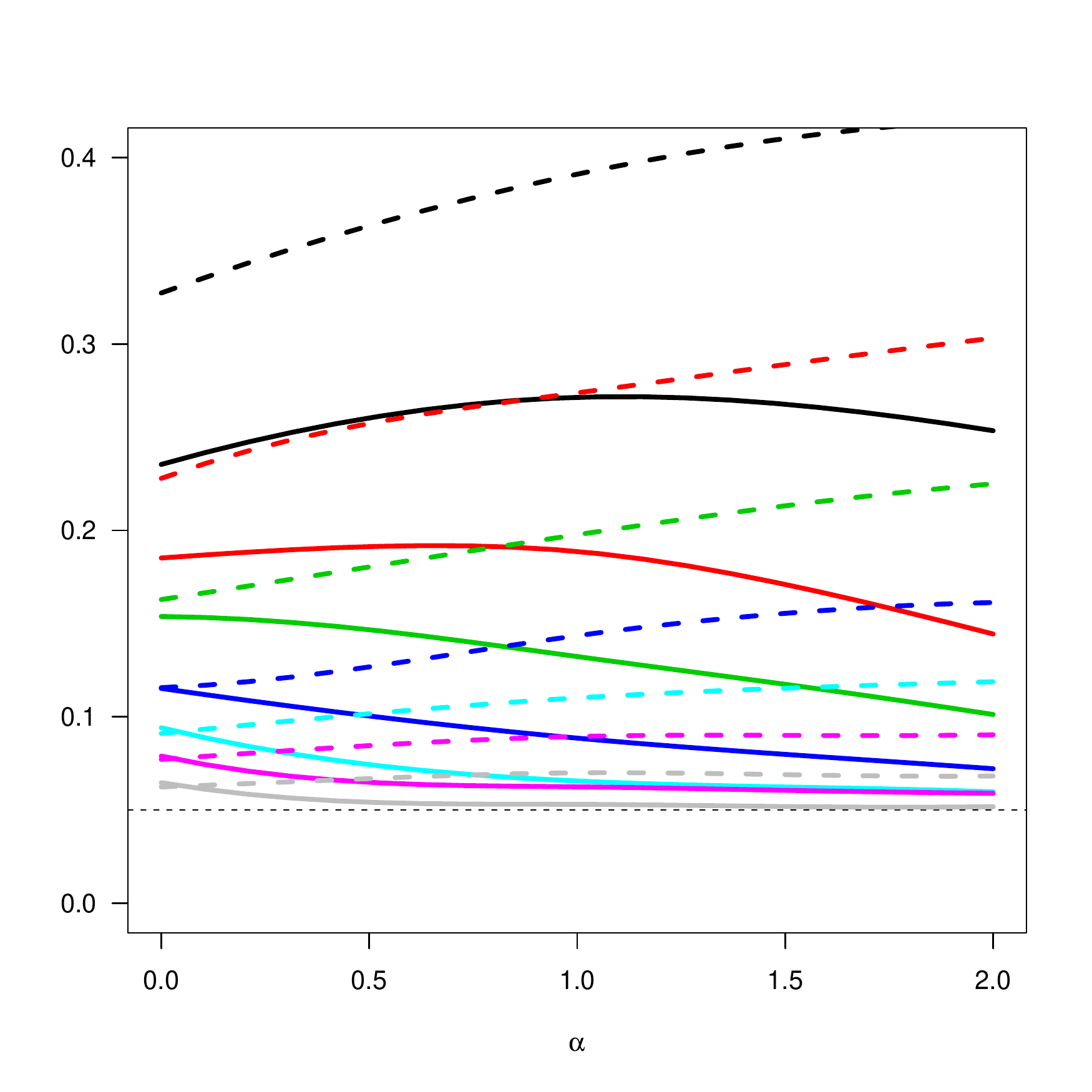}

}

\caption{Rejection frequencies for asymptotic $t$ tests calculated with asymptotic
standard errors: $\ell$-MVR and \textit{e}-MVR (solid lines), and
OLS and WLS (dashed lines). Sample sizes: 20 (black), 40 (red), 80
(green), 160 (blue), 320 (cyan), 640 (magenta), 1280 (grey). \label{fig:RejProbMVR1_HC0}}
\end{figure}

In order to further investigate the relative performance of MVR-based
inference, Tables \ref{tab:CIlength_OLS1} and \ref{tab:CIlength_WLS1}
report the ratio of average MVR confidence interval lengths across
simulations over the average OLS and WLS confidence interval lengths
for $\beta_{4}$ for each sample size and value of heteroskedasticity
index $\alpha$, in percentage terms. We find that in the presence
of heteroskedasticity, the length of MVR confidence intervals is shorter
for all designs compared to both OLS and WLS confidence intervals
for some $n$ large enough. The only exception is relative to $\ell$-WLS
for $\alpha=0.5,1$, as expected for $\alpha=1$ from $\ell$-WLS
being efficient in that case. The relatively larger average length
of the confidence intervals for $\ell$-MVR when $n=20$ in Tables
\ref{tab:CIlength_OLS1} and \ref{tab:CIlength_WLS1} is very much
reduced with finite-sample corrections (Tables 1-4 in the Supplementary
Material).

\begin{table}[t]
\centering %
\begin{tabular}{ccccccccccccc}
\toprule 
 &  & \multicolumn{5}{c}{$\ell$-MVR} &  & \multicolumn{5}{c}{\textit{e}-MVR}\tabularnewline
\midrule 
$\alpha$ &  & $0$ & $0.5$ & $1$ & $1.5$ & $2$ &  & $0$ & $0.5$ & $1$ & $1.5$ & $2$\tabularnewline
\midrule
$n=20$ &  & 204.1  & 193.8  & 179.9  & 167.2  & 158.7  &  & 114.7  & 121.8  & 123.2  & 123.4  & 121.9 \tabularnewline
$n=40$ &  & 121.0  & 117.5  & 109.8  & 100.0  & 91.0  &  & 104.2  & 109.8  & 108.8  & 103.1  & 94.1 \tabularnewline
$n=80$ &  & 106.7  & 105.9  & 97.7  & 85.5  & 75.0  &  & 101.6  & 106.0  & 100.7  & 89.1  & 74.3 \tabularnewline
$n=160$ &  & 102.9  & 101.8  & 90.1  & 74.8  & 64.0  &  & 100.9  & 103.3  & 92.9  & 76.4  & 59.1 \tabularnewline
$n=320$ &  & 101.4  & 98.5  & 83.1  & 65.6  & 55.3  &  & 100.6  & 100.0  & 84.8  & 65.1  & 47.3 \tabularnewline
$n=640$ &  & 100.6  & 95.2  & 76.2  & 57.5  & 47.8  &  & 100.2  & 96.3  & 77.0  & 55.6  & 38.2 \tabularnewline
$n=1280$ &  & 100.4  & 92.1  & 70.1  & 50.5  & 41.3  &  & 100.2  & 93.0  & 70.3  & 47.9  & 31.1 \tabularnewline
\bottomrule
\end{tabular}
\begin{centering}
\medskip{}
\par\end{centering}
\caption{Ratio ($\times100$) of MVR average confidence interval lengths for
$\beta_{4}$ over corresponding OLS counterpart. Confidence intervals
constructed with asymptotic standard errors.\label{tab:CIlength_OLS1}}
\bigskip{}
\centering %
\begin{tabular}{ccccccccccccc}
\toprule 
 &  & \multicolumn{5}{c}{$\ell$-MVR} &  & \multicolumn{5}{c}{\textit{e}-MVR}\tabularnewline
\midrule 
$\alpha$ &  & $0$ & $0.5$ & $1$ & $1.5$ & $2$ &  & $0$ & $0.5$ & $1$ & $1.5$ & $2$\tabularnewline
\midrule
$n=20$ &  & 245.2  & 226.0  & 210.7  & 201.3  & 197.4  &  & 136.1  & 143.3  & 149.8  & 158.7  & 165.9 \tabularnewline
$n=40$ &  & 140.4  & 128.8  & 123.4  & 120.1  & 117.0  &  & 113.0  & 118.8  & 126.1  & 132.7  & 133.8 \tabularnewline
$n=80$ &  & 120.5  & 111.5  & 109.8  & 107.5  & 103.2  &  & 105.7  & 110.7  & 116.1  & 119.4  & 113.5 \tabularnewline
$n=160$ &  & 112.7  & 105.7  & 106.3  & 101.9  & 96.7  &  & 102.8  & 106.8  & 108.6  & 107.2  & 96.1 \tabularnewline
$n=320$ &  & 108.7  & 103.4  & 105.4  & 98.3  & 92.7  &  & 101.4  & 103.9  & 102.0  & 96.5  & 81.8 \tabularnewline
$n=640$ &  & 105.5  & 102.7  & 105.1  & 95.3  & 89.0  &  & 100.7  & 101.2  & 95.8  & 87.0  & 70.1 \tabularnewline
$n=1280$ &  & 103.4  & 102.7  & 104.9  & 93.0  & 85.4  &  & 100.4  & 98.9  & 90.5  & 79.4  & 60.9 \tabularnewline
\bottomrule
\end{tabular}
\begin{centering}
\medskip{}
\par\end{centering}
\caption{Ratio ($\times100$) of MVR average confidence interval lengths for
$\beta_{4}$ over corresponding WLS counterpart. Confidence intervals
constructed with asymptotic standard errors.\label{tab:CIlength_WLS1}}
\end{table}

Overall these simulation results  demonstrate that MVR achieves large
improvements in terms of estimation and inference compared to OLS
in the presence of heteroskedasticity, and compared to WLS when the
conditional variance function is misspecified. Our numerical simulations
confirm MVR robustness to the specification of the scale function,
and both $\ell$-MVR and $e$-MVR perform very well in finite samples.
In the presence of heteroskedasticity rejection probabilities for
MVR are much closer to nominal level than those for OLS and for WLS
with misspecified weights. MVR achieves these improvements while simultaneously
displaying tighter confidence intervals in all designs for sample
sizes large enough. They are also shorter than their WLS counterpart
with misspecified weights for sample sizes large enough. The precision
of MVR estimates measured in RMSE is also largely superior to OLS
under heteroskedasticity and to WLS with misspecified weights, with
lower losses than WLS relative to OLS under homoskedasticity. These
results and the simulations in the Supplementary Material illustrate
the higher precision, improved finite-sample inference, and favorable
approximation properties of MVR compared to classical least-squares
methods.

\section{Conclusion}

We introduce a new loss function for the linear estimation and approximation
of CMFs. The proposed alternative generalises OLS, resulting in more
robust approximations under misspecification and large improvements
in finite samples. Given the importance of the least-squares loss
in econometrics and statistics, and the common occurence of heteroskedasticity
in empirical practice, the range of applications for simultaneous
mean-variance regression will be vast. Examples of natural avenues
for future research include the method of instrumental variables,
GARCH models, and flexible specification of the conditional variance
function for efficient estimation. These extensions will be considered
in subsequent work.

\appendix

\section{Theory for the MVR Criterion}

\subsection{Notation and Definitions}

We define
\[
L(X,Y,\theta):=\frac{1}{2}\left\{ \left(\frac{Y-X'\beta}{s(X'\gamma)}\right)^{2}+1\right\} s(X'\gamma),
\]
and
\[
\widetilde{L}(X,\theta):=\frac{1}{2}\left\{ \frac{E[(Y-X'\beta)^{2}\mid X]}{s(X'\gamma)}+s(X'\gamma)\right\} .
\]
so that by iterated expectations the objective function  can be expressed
as
\[
Q(\theta)=E[L(X,Y,\theta)]=E[\widetilde{L}(X,\theta)],\qquad\theta\in\Theta.
\]
We denote the level sets of $Q(\theta)$ by $\mathcal{B}_{c}=\{\theta\in\Theta:\;Q(\theta)\leq c\}$,
$c\in\mathbb{R}$, with boundary set $\partial\mathcal{B}_{c}$. We
also define the compact set $\mathcal{B}=\mathcal{B}_{\beta}\times\mathcal{B}_{\gamma}\subseteq\Theta$,
where $\mathcal{B}_{\beta}$ and $\mathcal{B}_{\gamma}$ are compact
subsets of $\mathbb{R}^{k}$ and $\Theta_{\gamma}$, respectively,
and the boundary set of $\Theta$
\[
\partial\Theta=\mathbb{R}^{k}\times\partial\Theta_{\gamma},\quad\partial\Theta_{\gamma}=\left\{ \gamma\in\mathbb{R}^{k}\,:\,\Pr[s(X'\gamma)=0]>0\right\} .
\]
For any two real numbers $a$ and $b$, $a\lor b=\max(a,b)$. For
two random variables $U$ and $V$, $\mathcal{U}$ denotes the support
of $U$, defined as the set of values of $U$ such that the density
$f_{U}(u)$ of $U$ is bounded away from $0$, and $\mathcal{U}_{v}$
is the conditional support of $U$ given $V=v$, $v\in\mathcal{V}$.
Throughout, $C$ is a generic constant whose value may change from
place to place.

\subsection{Preliminary Results}

This section gathers two preliminary results used in establishing
the properties of $Q(\theta)$.
\begin{lem}
Let $V$ be a random $k$ vector such that $E[VV']$ exists and is
nonsingular. Then, for every sequence $(\gamma_{n})$ in $\mathbb{R}^{k}$
such that $||\gamma_{n}||\rightarrow\infty$, there exists $v^{*}\in\mathcal{V}$
such that $\lim_{||\gamma_{n}||\rightarrow\infty}|\gamma_{n}'v^{*}|=\infty$
a.s.\label{lem:linindpdce}
\end{lem}
\begin{proof}
Consider a sequence $(\gamma_{n})$ in $\mathbb{R}^{k}$ such that
$||\gamma_{n}||\rightarrow\infty$, and define $\delta_{n}=\frac{\gamma_{n}}{||\gamma_{n}||}$.
The sequence $(\delta_{n})$ is bounded, and by application of the
Bolzano-Weierstrass theorem there exists a convergent subsequence
$\delta_{n_{l}}$, $n_{l}\rightarrow\infty$ as $l\rightarrow\infty$,
with limit $\delta_{o}$. Moreover, $E[VV']$ nonsingular implies
that it is positive definite, so that $E[(V'\delta_{o})^{2}]=\delta_{o}'E[VV']\delta_{o}>0$.
It follows that $V'\delta_{o}\neq0$ on a set of positive probability,
and there exists a value $v^{*}\in\mathcal{V}$ such that $\delta_{o}'v^{*}\neq0$
a.s. Therefore, $\delta_{n_{l}}=\frac{\gamma_{n_{l}}}{||\gamma_{n_{l}}||}$
satisfies $\delta_{n_{l}}'v^{*}\rightarrow\delta_{o}'v^{*}\neq0$
as $l\rightarrow\infty$, which implies that $\lim_{l\rightarrow\infty}|\gamma_{n_{l}}'v^{*}|\rightarrow\infty$:
\[
\lim_{l\rightarrow\infty}|\gamma_{n_{l}}'v^{*}|=\lim_{l\rightarrow\infty}\left|(\delta_{n_{l}}'v^{*})||\gamma_{n_{l}}||\right|=\left|(\delta_{o}'v^{*})\lim_{l\rightarrow\infty}||\gamma_{n_{l}}||\right|=\infty.
\]
The stated result follows.
\end{proof}
\begin{lem}
\label{thm:posdef}Suppose that Assumptions \ref{ass:ScaleSpec},
\ref{ass:Var(Y|X)} and \ref{ass:FullRank} hold. Then the matrix
\[
\varPsi(\theta)=E\left[\begin{array}{cc}
\frac{XX'}{s(X'\gamma)} & \frac{XX'}{s(X'\gamma)}s_{1}(X'\gamma)e(Y,X,\theta)\\
\frac{XX'}{s(X'\gamma)}s_{1}(X'\gamma)e(Y,X,\theta) & \frac{XX^{\top}}{s(\gamma\cdot X)}\{s_{1}(X'\gamma)e(Y,X,\theta)\}^{2}
\end{array}\right],
\]
defined for all $\theta\in\mathcal{B}$, is positive definite.
\end{lem}
\begin{proof}
The proof builds on the proofs of Lemma S3 in \citet{SS:2018} and
Theorem 1 in \citet{NS:2018}. Positive definiteness of $E[XX'/s(X'\gamma)]$
for $\gamma\in\mathcal{B}_{\gamma}$ under Assumption \ref{ass:FullRank}
implies that $\varPsi(\theta)$ is positive definite for all $\theta\in\mathcal{B}$
if and only if the Schur complement of $E[XX'/s(X'\gamma)]$ in $\varPsi(\theta)$
is positive definite (\citealp{BV:2004}, Appendix A.5.5) for all
$\theta\in\mathcal{B}$, i.e., if and only if
\begin{align*}
\varUpsilon(\theta) & :=E\left[\frac{XX'}{s(X'\gamma)}\{s_{1}(X'\gamma)e(Y,X,\theta)\}^{2}\right]\\
 & -E\left[\frac{XX'}{s(X'\gamma)}s_{1}(X'\gamma)e(Y,X,\theta)\right]E\left[\frac{XX'}{s(X'\gamma)}\right]^{-1}E\left[\frac{XX'}{s(X'\gamma)}s_{1}(X'\gamma)e(Y,X,\theta)\right]
\end{align*}
satisfies $\textrm{det}\{\varUpsilon(\theta)\}>0$, for all $\theta\in\mathcal{B}$.

Letting 
\[
\Xi(\theta):=E\left[\frac{XX'}{s(X'\gamma)}s_{1}(X'\gamma)e(Y,X,\theta)\right]E\left[\frac{XX'}{s(X'\gamma)}\right]^{-1},
\]
for all $\theta\in\mathcal{B}$, $\varUpsilon(\theta)$ is equal to
\[
E\left[\left\{ \frac{Xs_{1}(X'\gamma)e(Y,X,\theta)}{s(X'\gamma)^{1/2}}-\Xi(\theta)\frac{X}{s(X'\gamma)^{1/2}}\right\} \left\{ \frac{Xs_{1}(X'\gamma)e(Y,X,\theta)}{s(X'\gamma)^{1/2}}-\Xi(\theta)\frac{X}{s(X'\gamma)^{1/2}}\right\} ^{'}\right],
\]
a finite positive definite matrix, if and only if for all $\lambda\neq0$
and all $\theta\in\mathcal{B}$ there is no $d$ such that
\begin{equation}
\Pr\left[\left\{ \frac{\lambda'X}{s(X'\gamma)^{1/2}}\right\} s_{1}(X'\gamma)e(Y,X,\theta)=d'\left\{ \Xi(\theta)\frac{X}{s(X'\gamma)^{1/2}}\right\} \right]>0;\label{eq:posdefcond}
\end{equation}
this is an application of the Cauchy-Schwarz inequality for matrices
stated in \citet{Tripathi:1999}.

Positive definiteness of $E\left[XX'/s(X'\gamma)\right]$ for $\gamma\in\mathcal{B}_{\gamma}$
under Assumption \ref{ass:FullRank} implies that $E\left[\{\lambda'X\}^{2}/s(X'\gamma)\right]>0$
for all $\lambda\neq0$, which implies that $\Pr[\lambda'X/s(X'\gamma)^{1/2}\neq0]>0$
for all $\lambda\neq0$. Also, by Assumptions \ref{ass:ScaleSpec}-\ref{ass:Var(Y|X)},
we have $\Pr[s_{1}(X'\gamma)>0]=1$ and
\[
\Pr\left[\textrm{Var}(s_{1}(X'\gamma)e(Y,X,\theta)\mid X)>0\right]=\Pr\left[\left(\frac{s_{1}(X'\gamma)}{s(X'\gamma)}\right)^{2}\textrm{Var}(Y\mid X)>0\right]=1.
\]
Thus for all $\lambda\neq0$, by $\Xi(\theta)$ being a constant matrix
for all $\theta\in\mathcal{B}$, there is no $d$ such that (\ref{eq:posdefcond})
holds, and the result follows.
\end{proof}

\subsection{Main Properties of $Q(\theta)$}
\begin{lem}
\label{lem:QCont}{[}Continuity{]} Suppose that Assumptions \ref{ass:ScaleSpec}
and \ref{ass:Moments} hold. Then $\theta\mapsto Q(\theta)$ is continuous
over $\mathcal{B}$.
\end{lem}
\begin{proof}
We first show that $E[\sup_{\theta\in\mathcal{B}}\left|L(X,Y,\theta)\right|]<\infty$
for all $\theta\in\mathcal{B}$. By the Triangle Inequality,
\begin{equation}
2\left|L(X,Y,\theta)\right|\leq|e(Y,X,\theta)^{2}s(X'\gamma)|+|s(X'\gamma)|.\label{eq:TI}
\end{equation}
Compactness of $\mathcal{B}_{\gamma}$ implies that there exists a
constant $C$ such that $\sup_{\gamma\in\mathcal{B}_{\gamma}}1/s(X'\gamma)\leq C<\infty$
a.s. Thus for $\theta\in\mathcal{B}$, the bound
\begin{equation}
|e(Y,X,\theta)^{2}s(X'\gamma)|\leq C[2Y^{2}+2(X'\beta)^{2}]\leq2C[Y^{2}+\sup_{\beta\in\mathcal{B}_{\beta}}||\beta||^{2}||X||^{2}],\label{eq:Bound-1}
\end{equation}
and $\sup_{\beta\in\mathcal{B}_{\beta}}||\beta||<\infty$ together
imply that $E[\sup_{\beta\in\mathcal{B}}|e(Y,X,\theta)^{2}s(X'\gamma)|]<\infty$
requires $E[Y^{2}]<\infty$ and $E||X||^{2}<\infty$, which hold under
Assumption \ref{ass:Moments}(i).

It remains to show that $E[\sup_{\gamma\in\mathcal{B}_{\gamma}}|s(X'\gamma)|]<\infty$.
For $\gamma\in\mathcal{B}_{\gamma}$, some $0\leq\kappa\in(a,\infty)$
and some intermediate values $\bar{\gamma}$, a mean-value expansion
about $(\kappa,0_{k-1})'$ yields
\[
|s(X'\gamma)|=|s(\kappa)+s_{1}(X'\bar{\gamma})(X'\gamma-\kappa)|\leq s(\kappa)+\sup_{\gamma\in\mathcal{B}_{\gamma}}||\gamma||\,||X||s_{1}(X'\bar{\gamma}).
\]
With $s(\kappa),\sup_{\gamma\in\mathcal{B}_{\gamma}}||\gamma||<\infty$,
$E[\sup_{\gamma\in\mathcal{B}_{\gamma}}|s(X'\gamma)|]<\infty$ requires
$E[||X||s_{1}(X'\bar{\gamma})]<\infty$, which holds under Assumption
\ref{ass:Moments}(ii). 

Bound (\ref{eq:TI}) now implies that $E[\sup_{\theta\in\mathcal{B}}\left|L(X,Y,\theta)\right|]<\infty$,
and continuity of $Q(\theta)$ then follows from continuity of $\theta\mapsto L(X,Y,\theta)$
and dominated convergence.
\end{proof}
\begin{lem}
\label{lem:QDiff}{[}Continuous Differentiability{]} If Assumptions
\ref{ass:ScaleSpec} and \ref{ass:Moments} hold, then, for all $\theta\in\mathcal{B}$,
$Q(\theta)$ is continuously differentiable and $\partial E[L(X,Y,\theta)]/\partial\theta=E[\partial L(X,Y,\theta)/\partial\theta]$.
\end{lem}
\begin{proof}
\begin{sloppy}We first show that $E[\sup_{\theta\in\mathcal{B}}||\partial L(X,Y,\theta)/\partial\theta||]<\infty$.
Computing
\[
\partial L(X,Y,\theta)/\partial\beta=-Xe(Y,X,\theta),\quad\partial L(X,Y,\theta)/\partial\gamma=-\frac{1}{2}Xs_{1}(X'\gamma)\{e(Y,X,\theta)^{2}-1\}.
\]
Compactness of $\mathcal{B}_{\gamma}$ implies that there exists a
constant $C$ such that $\sup_{\gamma\in\mathcal{B}_{\gamma}}1/s(X'\gamma)\leq C<\infty$
a.s. Thus for $\theta\in\mathcal{B}$, the bound
\[
||Xe(Y,X,\theta)||\leq C||X||\,|Y-X'\beta|\leq C[|Y|\,||X||+\sup_{\beta\in\mathcal{B}_{\beta}}||\beta||\,||X||^{2}],
\]
and $\sup_{\beta\in\mathcal{B}_{\beta}}||\beta||<\infty$, imply that
$E[\sup_{\theta\in\mathcal{B}}||\partial L(X,Y,\theta)/\partial\beta||]<\infty$
requires $E[|Y|\,||X||]<\infty$ and $E||X||^{2}<\infty$, which hold
under Assumptions \ref{ass:Moments}(i).\par\end{sloppy}

We now show that $E[\sup_{\theta\in\mathcal{B}}||\partial L(X,Y,\theta)/\partial\gamma||]<\infty$.
Since $-1\leq e(Y,X,\theta)^{2}-1$ a.s., for $\theta\in\mathcal{B}$,
\begin{eqnarray}
\left\Vert Xs_{1}(X'\gamma)\{e(Y,X,\theta)^{2}-1\}\right\Vert  & \leq & ||X||\,s_{1}(X'\gamma)\,|e(Y,X,\theta)^{2}-1|\nonumber \\
 & \leq & ||X||\,s_{1}(X'\gamma)\,[1\vee2C(Y^{2}+\sup_{\beta\in\mathcal{B}_{\beta}}||\beta||^{2}||X||^{2})].\label{eq:Bound}
\end{eqnarray}
For $\gamma\in\mathcal{B}_{\gamma}$, some $0\leq\kappa\in(a,\infty)$
and some intermediate values $\bar{\gamma}$, a mean-value expansion
about $(\kappa,0_{k-1})'$ yields
\begin{equation}
|s_{1}(X'\gamma)|=|s_{1}(\kappa)+s_{2}(X'\bar{\gamma})(X'\gamma-\kappa)|\leq s_{1}(\kappa)+\sup_{\gamma\in\mathcal{B}_{\gamma}}||\gamma||\,||X||s_{2}(X'\bar{\gamma}).\label{eq:sprimeMVexpansion}
\end{equation}
This bound and (\ref{eq:Bound}) together imply
\begin{eqnarray}
\left\Vert Xs_{1}(X'\gamma)\{e(Y,X,\theta)^{2}-1\}\right\Vert  & \leq & ||X||\,[s_{1}(\kappa)+\sup_{\gamma\in\mathcal{B}_{\gamma}}||\gamma||\,||X||s_{2}(X'\bar{\gamma})]\nonumber \\
 &  & \times[1\vee2C(Y^{2}+\sup_{\beta\in\mathcal{B}_{\beta}}||\beta||^{2}||X||^{2})].\label{eq:BoundExp}
\end{eqnarray}
Since $\mathcal{B}$ is compact and $0<s_{1}(\kappa)<\infty$, $E[\sup_{\theta\in\mathcal{B}}||\nabla_{\gamma}L(X,Y,\theta)||]<\infty$
requires $E||X||^{3}<\infty$, $E[Y^{2}\,||X||]<\infty$ and, for
all $\gamma\in\mathcal{B}_{\gamma}$, $E[\left\Vert X\right\Vert ^{4}s_{2}(X'\gamma)]<\infty$
and $E[Y^{2}||X||^{2}s_{2}(X'\gamma)]<\infty$, which hold under Assumption
\ref{ass:Moments}. 

We have shown that $E[\sup_{\theta\in\mathcal{B}}||\partial L(X,Y,\theta)/\partial\theta||]<\infty$
and it now follows by Lemma 3.6 in \citet{Newey:McFadden:1994} that
$Q(\theta)$ is continuously differentiable over $\mathcal{B}$, and
that the order of differentiation and integration can be interchanged.
\end{proof}
\begin{lem}
\label{lem:Convexity}{[}Convexity{]} Suppose that Assumptions \ref{ass:ScaleSpec},
\ref{ass:Moments} and \ref{ass:FullRank} hold. Then $\theta\mapsto Q(\theta)$
is strictly convex over $\mathcal{B}$.
\end{lem}
\begin{proof}
$Q(\theta)$ is differentiable for all $\theta\in\mathcal{B}$ and
the order of integration and differentiation can be interchanged by
Lemma \ref{lem:QDiff}. In order to show that $\partial Q(\theta)/\partial\theta$
is differentiable for $\theta\in\mathcal{B}$, we show that $E[\sup_{\theta\in\mathcal{B}}||\partial^{2}L(X,Y,\theta)/\partial\theta\partial\theta||]<\infty$.
Direct calculations yield
\begin{eqnarray*}
\frac{\partial^{2}L(X,Y,\theta)}{\partial\theta\partial\theta} & = & \left[\begin{array}{cc}
\frac{XX'}{s(X'\gamma)} & \frac{XX'}{s(X'\gamma)}s_{1}(X'\gamma)e(Y,X,\theta)\\
\frac{XX'}{s(X'\gamma)}s_{1}(X'\gamma)e(Y,X,\theta) & \frac{XX'}{s(X'\gamma)}\{s_{1}(X'\gamma)e(Y,X,\theta)\}^{2}
\end{array}\right]\\
 &  & +\left[\begin{array}{cc}
0_{k\times k} & 0_{k\times k}\\
0_{k\times k} & -\frac{1}{2}XX's_{2}(X'\gamma)\{e(Y,X,\theta)^{2}-1\}
\end{array}\right]\\
 & =: & h_{1}(X,Y,\theta)+h_{2}(X,Y,\theta).
\end{eqnarray*}
\begin{sloppy}We first consider $E[\sup_{\theta\in\mathcal{B}}||\partial^{2}h_{1}(X,Y,\theta)/\partial\theta\partial\theta||]$.
Steps similar to those leading to (\ref{eq:Bound-1}) imply that $E[\sup_{\theta\in\mathcal{B}}||\partial^{2}h_{1}(X,Y,\theta)/\partial\beta\partial\beta||]<\infty$
is satisfied since $E||X||^{2}<\infty$ and $\mathcal{B}$ is compact.
Moreover, $E[\sup_{\theta\in\mathcal{B}}||\partial^{2}h_{1}(X,Y,\theta)/\partial\beta\partial\gamma||]<\infty$
and $E[\sup_{\theta\in\mathcal{B}}||\partial^{2}h_{1}(X,Y,\theta)/\partial\gamma\partial\beta||]<\infty$
are implied by $E[\sup_{\theta\in\mathcal{B}}||\partial^{2}h_{1}(X,Y,\theta)/\partial\gamma\partial\gamma||]<\infty$. 

Steps similar to those leading to (\ref{eq:BoundExp}) yield, for
$\theta\in\mathcal{B}$,
\begin{align*}
\left\Vert \frac{XX'}{s(X'\gamma)}\{s_{1}(X'\gamma)e(Y,X,\theta)\}^{2}\right\Vert  & \leq C||X||^{2}s_{1}(X'\gamma)^{2}[Y^{2}+\sup_{\beta\in\mathcal{B}_{\beta}}||\beta||^{2}||X||^{2}].
\end{align*}
This bound and expansion (\ref{eq:sprimeMVexpansion}) together imply,
for some $0\leq\kappa\in(a,\infty)$ and some intermediate value $\bar{\gamma}$,
\begin{align*}
\left\Vert \frac{XX'}{s(X'\gamma)}\{s_{1}(X'\gamma)e(Y,X,\theta)\}^{2}\right\Vert  & \leq C||X||^{2}\,[s_{1}(\kappa)^{2}+2\sup_{\gamma\in\mathcal{B}_{\gamma}}||\gamma||\,||X||s_{1}(X'\bar{\gamma})s_{2}(X'\bar{\gamma})]\\
 & \quad\times[Y^{2}+\sup_{\beta\in\mathcal{B}_{\beta}}||\beta||^{2}||X||^{2}].
\end{align*}
Since $\mathcal{B}$ is compact and $0<s_{1}(\kappa)<\infty$, $E[\sup_{\theta\in\mathcal{B}}||\partial^{2}h_{1}(X,Y,\theta)/\partial\gamma\partial\gamma||]<\infty$
requires $E||X||^{4}<\infty$, $E[Y^{2}\,||X||^{2}]<\infty$ and,
for all $\gamma\in\mathcal{B}_{\gamma}$, $E[\left\Vert X\right\Vert ^{5}s_{1}(X'\gamma)s_{2}(X'\gamma)]<\infty$
and $E\left[Y^{2}||X||^{3}s_{1}(X'\gamma)s_{2}(X'\gamma)\right]<\infty$,
which hold under Assumption \ref{ass:Moments}.

We now show that $E[\sup_{\theta\in\mathcal{B}}||\partial^{2}h_{2}(X,Y,\theta)/\partial\theta\partial\theta||]<\infty$.
It suffices to show that $E[\sup_{\theta\in\mathcal{B}}||\partial^{2}h_{2}(X,Y,\theta)/\partial\gamma\partial\gamma||]<\infty$.
Steps similar to those leading to (\ref{eq:BoundExp}), yield, for
$\theta\in\mathcal{B}$,
\begin{align}
\left\Vert XX's_{2}(X'\gamma)\{e(Y,X,\theta)^{2}-1\}\right\Vert  & \leq||X||^{2}s_{2}(X'\gamma)[1\vee2C(Y^{2}+\sup_{\beta\in\mathcal{B}_{\beta}}||\beta||^{2}||X||^{2})]].\label{eq:Bound2}
\end{align}
For $\gamma\in\mathcal{B}_{\gamma}$ a mean-value expansion about
$(\kappa,0_{k-1})'$ yields
\[
|s_{2}(X'\gamma)|=|s_{2}(\kappa)+s_{3}(X'\bar{\gamma})(X'\gamma-\kappa)|\leq s_{2}(\kappa)+\sup_{\gamma\in\mathcal{B}_{\gamma}}||\gamma||\,||X||s_{3}(X'\bar{\gamma})
\]
This bound and (\ref{eq:Bound2}) together imply
\begin{eqnarray*}
\left\Vert XX's_{2}(X'\gamma)\{e(Y,X,\theta)^{2}-1\}\right\Vert  & \leq & ||X||^{2}\,[s_{2}(\kappa)+\sup_{\gamma\in\mathcal{B}_{\gamma}}||\gamma||\,||X||s_{3}(X'\bar{\gamma})]\\
 &  & \times[1\vee2C(Y^{2}+\sup_{\beta\in\mathcal{B}_{\beta}}||\beta||^{2}||X||^{2})].
\end{eqnarray*}
Since $\mathcal{B}$ is compact and $0\leq s_{2}(\kappa)<\infty$,
$E[\sup_{\theta\in\mathcal{B}}||\partial^{2}h_{2}(X,Y,\theta)/\partial\gamma\partial\gamma||]<\infty$
requires $E||X||^{4}<\infty$, $E[Y^{2}\,||X||^{2}]<\infty$ and,
for all $\gamma\in\mathcal{B}_{\gamma}$, $E[\left\Vert X\right\Vert ^{5}s_{3}(X'\gamma)]<\infty$
and $E[Y^{2}||X||^{3}s_{3}(X'\gamma)]<\infty$, which hold under Assumption
\ref{ass:Moments}. 

We have shown that $E[\sup_{\theta\in\mathcal{B}}||\partial^{2}L(X,Y,\theta)/\partial\theta\partial\theta||]<\infty$
and it follows by Lemma 3.6 in \citet{Newey:McFadden:1994} that $\partial Q(\theta)/\partial\theta$
is continuously differentiable over $\mathcal{B}$, and that the order
of differentiation and integration can be interchanged.\par\end{sloppy}

Letting $H_{1}(\theta):=E[h_{1}(X,Y,\theta)]$ and $H_{2}(\theta):=E[h_{2}(X,Y,\theta)]$,
for all $\theta\in\mathcal{B}$, the Hessian matrix of $Q(\theta)$
is $H(\theta):=H_{1}(\theta)+H_{2}(\theta)$, which is positive semidefinite
if $H_{1}(\theta)$ and $H_{2}(\theta)$ are positive semidefinite
(\citealp{HJ:2012}, p.398, 7.1.3. observation). And if either one
of $H_{1}(\theta)$ and $H_{2}(\theta)$ is positive definite (while
the other is positive semidefinite), then $H(\theta)$ is positive
definite. All principal minors of $H_{2}(\theta)$ have determinant
$0$ for all $\theta\in\mathcal{B}$, and $H_{2}(\theta)$ is thus
positive semidefinite. Applying Lemma \ref{thm:posdef} with $\varPsi(\theta)=H_{1}(\theta)$,
we have that $H_{1}(\theta)$ is positive definite for all $\theta\in\mathcal{B}$.
We conclude that $H(\theta)$ is positive definite for all $\theta\in\mathcal{B}$,
and the result follows.
\end{proof}
\begin{lem}
\label{lem:CompactLevSets}{[}Level Sets Compactness{]} If Assumptions
\ref{ass:ScaleSpec}, \ref{ass:Moments} and \ref{ass:FullRank} hold
then the level sets of $\theta\mapsto Q(\theta)$ are compact.
\end{lem}
\begin{proof}
We show that the level sets $\mathcal{B}_{c}$, $c\in\mathbb{R}$,
of $\theta\mapsto Q(\theta)$ are closed and bounded. The result then
follows by the Heine-Borel theorem.

Step 1. \textit{{[}$\mathcal{B}_{c}$ is bounded{]}}. We show that
every sequence in $\mathcal{B}_{c}$ is bounded. Suppose the contrary.
Then there exists an unbounded sequence $(\theta_{n})$ in $\mathcal{B}_{c}$,
and a subsequence $(\theta_{n_{l}})$, $n_{l}\rightarrow\infty$ as
$l\rightarrow\infty$, such that either $||\beta_{n_{l}}||\rightarrow\infty$
or $||\gamma_{n_{l}}||\rightarrow\infty$.

Step 1.1. If $||\gamma_{n_{l}}||\rightarrow\infty$, then $E[XX']$
nonsingular implies that there exists a value $x^{*}\in\mathcal{X}$
such that $|\gamma_{n_{l}}'x^{*}|\rightarrow\infty$ as $l\rightarrow\infty$,
a.s., by Lemma \ref{lem:linindpdce}, which implies $s(\gamma_{n_{l}}'x^{*})\rightarrow0$
or $\infty$ by definition of $t\mapsto s(t)$ in Assumption \ref{ass:ScaleSpec}.

Moreover, for $x^{*}\in\mathcal{X}$ such that $|\gamma_{n_{l}}'x^{*}|\rightarrow\infty$
as $l\rightarrow\infty$, we have that $E[(Y-X'\beta)^{2}\mid X=x^{*}]<\infty$
for all $\beta\in\mathbb{R}^{k}$ under Assumption \ref{ass:Moments}(i).
It follows that for all $\beta\in\mathbb{R}^{k}$,
\begin{equation}
\lim_{l\rightarrow\infty}\tilde{L}(x^{*},\beta,\gamma_{n_{l}})=\frac{1}{2}\lim_{l\rightarrow\infty}\frac{E[(Y-X'\beta)^{2}\mid X=x^{*}]}{s(\gamma_{n_{l}}'x^{*})}+\frac{1}{2}\lim_{l\rightarrow\infty}s(\gamma_{n_{l}}'x^{*})=\infty.\label{eq:lim}
\end{equation}
Since $\tilde{L}(X,\theta)$ is positive a.s. for all $\theta\in\Theta$
and the density $f_{X}(x)$ is bounded away from 0 for all $x\in\mathcal{X}$
by definition of $\mathcal{X}$, (\ref{eq:lim}) implies that $E[\lim_{l\rightarrow\infty}\tilde{L}(X,\beta,\gamma_{n_{l}})]=\infty$.
Since $E[\tilde{L}(X,\theta)]=Q(\theta)$, Fatou's lemma then implies
that $\lim_{l\rightarrow\infty}Q(\beta,\gamma_{n_{l}})=\infty$, for
all $\beta\in\mathbb{R}^{k}$. Therefore $\gamma$ is bounded.

Step 1.2. If $||\beta_{n_{l}}||\rightarrow\infty$, then $E[XX']$
nonsingular implies that there exists a value $x^{**}\in\mathcal{X}$
such that $|\beta_{n_{l}}'x^{**}|\rightarrow\infty$ a.s., by a second
application of Lemma \ref{lem:linindpdce}. Thus $(Y-\beta_{n_{l}}'x^{**})^{2}\rightarrow\infty$
as $l\rightarrow\infty$ a.s.

Moreover, for $x^{**}\in\mathcal{X}$ such that $|\beta_{n_{l}}'x^{**}|\rightarrow\infty$
as $l\rightarrow\infty$, we have that $E[\lim_{l\rightarrow\infty}(Y-X'\beta_{n_{l}})^{2}\mid X=x^{**}]=\infty$,
and Fatou's lemma then implies that $\lim_{l\rightarrow\infty}E[(Y-X'\beta_{n_{l}})^{2}\mid X=x^{**}]=\infty$.
Also, $s(\gamma'x^{**})$ is finite and positive for any $\gamma\in\Theta_{\gamma}$.
It follows that for all $\gamma\in\Theta_{\gamma}$, 
\begin{eqnarray}
\lim_{l\rightarrow\infty}\tilde{L}(x^{**},\beta_{n_{l}},\gamma) & \geq & \frac{1}{2}\lim_{l\rightarrow\infty}\frac{E[(Y-X'\beta_{n_{l}})^{2}\mid X=x^{**}]}{s(\gamma'x^{**})}=\infty.\label{eq:lim-1}
\end{eqnarray}
Since $\tilde{L}(X,\theta)$ is positive a.s. for all $\theta\in\Theta$
and the density $f_{X}(x)$ is bounded away from 0 for all $x\in\mathcal{X}$
by definition of $\mathcal{X}$, (\ref{eq:lim-1}) implies that $E[\lim_{l\rightarrow\infty}\tilde{L}(X,\beta_{n_{l}},\gamma)]=\infty$.
Fatou's lemma then implies that $\lim_{l\rightarrow\infty}E[\tilde{L}(X,\beta_{n_{l}},\gamma)]=\lim_{l\rightarrow\infty}Q(\beta_{n_{l}},\gamma)=\infty$
for all $\gamma\in\Theta_{\gamma}$. Therefore $\beta$ is bounded.

Step 2.\textbf{ }\textit{{[}$\mathcal{B}_{c}$ is closed{]}}. We examine
the behaviour of $\theta\mapsto Q(\theta)$ on the boundary set $\partial\Theta$
in order to determine whether $\mathcal{B}_{c}$ is closed. We show
that for every sequence in $\mathcal{B}_{c}$ converging to a boundary
point in $\partial\Theta$, $\theta\mapsto Q(\theta)$ is unbounded.
Continuity of $\theta\mapsto Q(\theta)$ established in Lemma \ref{lem:QCont}
then implies that $\mathcal{B}_{c}$ is closed.

Consider a sequence $\theta_{n}$ in $\mathcal{B}_{c}$ such that
$\theta_{n}\rightarrow t_{o}\in\partial\Theta$ as $n\rightarrow\infty$.
Then, there exists $x^{*}\in\mathcal{X}$ such that $\gamma_{n}'x^{*}\rightarrow0$
as $n\rightarrow\infty$ a.s., by definition of $\partial\Theta$.
Moreover, $E[(Y-X'\beta)^{2}|X]>0$ a.s. for all $\beta\in\mathbb{R}^{k}$
under Assumption \ref{ass:Var(Y|X)}. Thus
\begin{equation}
\lim_{n\rightarrow\infty}\tilde{L}(x^{*},\theta_{n})=\frac{1}{2}\lim_{n\rightarrow\infty}\frac{E[(Y-X'\beta_{n})^{2}\mid X=x^{*}]}{\gamma_{n}'x^{*}}=\infty.\label{eq:lim-2}
\end{equation}
Since $\tilde{L}(X,\theta)$ is positive a.s. for all $\theta\in\Theta$
and the density $f_{X}(x)$ is bounded away from 0 for all $x\in\mathcal{X}$
by definition of $\mathcal{X}$, (\ref{eq:lim-2}) implies that $E[\lim_{n\rightarrow\infty}\tilde{L}(X,\theta_{n})]=\infty$.
Fatou's lemma then implies that $\lim_{n\rightarrow\infty}E[\tilde{L}(X,\theta_{n})]=\lim_{n\rightarrow\infty}Q(\theta_{n})=\infty$.
This yields a contradiction since $Q(\theta_{n})\leq c$ for $\theta_{n}\in\mathcal{B}_{c}$.
Moreover, continuity of $\theta\mapsto Q(\theta)$ implies $Q(t_{o})=\lim_{n\rightarrow\infty}Q(\theta_{n})\leq c$.
Therefore, $t_{o}\in\mathcal{B}_{c}$ and $\mathcal{B}_{c}$ is closed.
\end{proof}

\section{Proofs for Sections \ref{sec:SMVR} and \ref{sec:Approximation-Properties}}

\subsection{Proof of Theorem \ref{thm:Uniqueness}}

Under Assumptions \ref{ass:ScaleSpec}-\ref{ass:FullRank}, the order
of integration and differentiation for the MVR population problem
(\ref{eq:MVR}) can be interchanged by Lemma \ref{lem:QDiff}. Therefore
the first-order conditions of problem (\ref{eq:MVR}) are (\ref{eq:FOC1})-(\ref{eq:FOC2}),
which are satisfied by $\theta_{0}$. Uniqueness follows from strict
convexity of $Q(\theta)$ over compact subsets of $\Theta$ established
in Lemma \ref{lem:Convexity}, and compactness of the level sets of
the objective function $Q(\theta)$ established in Lemma \ref{lem:CompactLevSets}.\qed

\subsection{Proof of Proposition \ref{prop:Equivalence}}

The first-order conditions (\ref{eq:FOC1})-(\ref{eq:FOC2}) of the
MVR population problem (\ref{eq:MVR}) can be written as 
\begin{eqnarray}
E\left[\frac{X}{s(X'\gamma)}(Y-X'\beta)\right] & = & 0\label{eq:FOC1_2}\\
E\left[X\frac{s_{1}(X'\gamma)}{s(X'\gamma)^{2}}\left\{ (Y-X'\beta)^{2}-s(X'\gamma)^{2}\right\} \right] & = & 0,\label{eq:FOC2_2}
\end{eqnarray}
with unique solutions $\beta_{0}$ and $\gamma_{0}$, by Theorem \ref{thm:Uniqueness}.

Under the stated assumptions, the first-order conditions of problem
(\ref{eq:InfeasibleWLS})-(\ref{eq:InfeasibleWNLS}) are
\begin{eqnarray}
E\left[\frac{X}{\sigma(X)}(Y-X'\beta)\right] & = & 0\label{eq:FOC_IWLS}\\
-4E\left[X\frac{s(X'\gamma)s_{1}(X'\gamma)}{\sigma(X)^{3}}\left\{ (Y-X'\beta_{0})^{2}-s(X'\gamma)^{2}\right\} \right] & = & 0.\label{eq:FOC_IWNLS}
\end{eqnarray}
By assumption the variance of $Y$ conditional on $X$ is correctly
specified and $\sigma(X)^{2}=s(X'\gamma_{0})^{2}$ a.s. Conditions
(\ref{eq:FOC_IWLS})-(\ref{eq:FOC_IWNLS}) are therefore satisfied
for $(\beta,\gamma)=(\beta_{0},\gamma_{0})$, and are then equivalent
to the MVR first-order conditions (\ref{eq:FOC1_2})-(\ref{eq:FOC2_2})
evaluated at the solution $(\beta,\gamma)=(\beta_{0},\gamma_{0})$.\qed

\subsection{Proof of Theorem \ref{thm:Existence}}

\textit{Step 1: Existence}. Pick $c\in\mathbb{R}$ such that the level
set $\mathcal{B}_{c}=\{\theta\in\Theta:\;Q(\theta)\leq c\}$ is nonempty.
By Lemma \ref{lem:CompactLevSets}, $\mathcal{B}_{c}$ is compact.
Continuity of $\theta\mapsto Q(\theta)$ over compact subsets of $\Theta$
established in Lemma \ref{lem:QCont} then implies that there exists
at least one minimizer to $Q(\theta)$ in $\mathcal{B}_{c}$ by the
Weierstrass extreme value theorem. Minimizing $Q(\theta)$ over $\Theta$
is equivalent to minimizing $Q(\theta)$ over any of its nonempty
level sets, which establishes existence of a minimizer $\theta^{*}\in\Theta$.

\textit{Step 2: Uniqueness}. By Lemma \ref{lem:Convexity}, we have
that $\theta\mapsto Q(\theta)$ is strictly convex over compact subsets
of $\Theta$, and thus over $\mathcal{B}_{c}$, so that $Q(\theta)$
admits at most one minimizer in $\mathcal{B}_{c}$. This establishes
uniqueness of a minimizer $\theta^{*}\in\Theta$. \qed

\subsection{Proof of Theorem \ref{thm:Interpretation}}

\textit{Proof of part (i)}. For $\theta\in\Theta$, define the function
\[
\widetilde{Q}(\theta):=\frac{1}{2}\int\left\{ \left(\frac{\mu(x)-x'\beta}{s(x'\gamma)}\right)^{2}+\frac{\sigma(x)^{2}}{s(x'\gamma){}^{2}}+1\right\} s(x'\gamma)dF_{X}(x).
\]
We show that $\widetilde{Q}(\theta)$ is equal to $Q(\theta)$ for
all $\theta\in\Theta$. 

The location-scale representation (\ref{eq:loc-scale rep}) for $Y\mid X$
 implies that, for $\theta\in\Theta$,
\begin{eqnarray}
\left(\frac{Y-X'\beta}{s(X'\gamma)}\right)^{2} & = & \left(\frac{[\mu(X)-X'\beta]+\sigma(X)\varepsilon}{s(X'\gamma)}\right)^{2}\nonumber \\
 & = & \left(\frac{\mu(X)-X'\beta}{s(X'\gamma)}\right)^{2}+2\left(\frac{\mu(X)-X'\beta}{s(X'\gamma)}\right)\frac{\sigma(X)}{s(X'\gamma)}\varepsilon+\frac{\sigma(X)^{2}}{s(X'\gamma)^{2}}\varepsilon^{2},\label{eq:substitution}
\end{eqnarray}
and the change of variable formula 
\begin{equation}
f_{Y\mid X}(Y\mid X)=f_{\varepsilon\mid X}\left(\frac{Y-\mu(X)}{\sigma(X)}\mid X\right)\left(\frac{1}{\sigma(X)}\right),\label{eq:chov}
\end{equation}
hold a.s. 

The definition of $Q(\theta)$ and expressions (\ref{eq:substitution})-(\ref{eq:chov})
together imply, for $\theta\in\Theta$,
\begin{eqnarray*}
Q(\theta) & = & \frac{1}{2}\int\left\{ \left(\frac{y-x'\beta}{s(x'\gamma)}\right)^{2}+1\right\} s(x'\gamma)f_{Y\mid X}(y\mid x)dydF_{X}(x)\\
 & = & \frac{1}{2}\int\left\{ \left(\frac{y-x'\beta}{s(x'\gamma)}\right)^{2}+1\right\} s(x'\gamma)f_{\varepsilon\mid X}\left(\frac{y-\mu(x)}{\sigma(x)}\mid x\right)\left(\frac{1}{\sigma(x)}\right)dydF_{X}(x)\\
 & = & \frac{1}{2}\int\left\{ \left(\frac{\mu(x)-x'\beta}{s(x'\gamma)}\right)^{2}+2\left(\frac{\mu(x)-x'\beta}{s(x'\gamma)}\right)\frac{\sigma(x)}{s(x'\gamma)}e+\frac{\sigma(x)^{2}}{s(x'\gamma)^{2}}e^{2}+1\right\} \\
 &  & \quad\times s(x'\gamma)f_{\varepsilon\mid X}\left(e\mid x\right)dedF_{X}(x)\\
 & = & \frac{1}{2}\int\left\{ \left(\frac{\mu(x)-x'\beta}{s(x'\gamma)}\right)^{2}+\frac{\sigma(x)^{2}}{s(x'\gamma)^{2}}+1\right\} s(x'\gamma)dF_{X}(x)=\widetilde{Q}(\theta),
\end{eqnarray*}
where the final step uses the law of iterated expectations and the
mean zero and unit variance property of $\varepsilon$ conditional
on $X$. Since $\theta^{*}$ is the unique minimizer of $Q(\theta)$
in $\Theta$, it is also the unique minimizer of $\widetilde{Q}(\theta)$
in $\Theta$.\qed

\textit{Proof of part (ii)}. Since $\Theta_{\textrm{LS}}\subset\Theta$,
and $\theta^{*}$ and $\theta_{\textrm{LS}}$ are the unique minimizers
of $Q(\theta)$ over $\Theta$ and $\Theta_{\textrm{LS}}$, respectively,
it follows that $Q(\theta^{*})\leq Q(\theta_{\textrm{LS}})$.\qed

\subsection{Proof of Corollary \ref{cor:Approx Bound}}

For the linear scale specification $s(t)=t$, conditions (\ref{eq:FOC2})
imply $E[(X'\gamma^{*})e(Y,X,\theta^{*})^{2}]=E[X'\gamma^{*}]$. For
the exponential scale specification $s(t)=\exp(t)$, because $X$
includes an intercept, conditions (\ref{eq:FOC2}) imply that $E[\exp(X'\gamma^{*})e(Y,X,\theta^{*})^{2}]=E[\exp(X'\gamma^{*})]$.
It follows that for the linear and exponential scale specifications,
$E[s(X'\gamma^{*})e(Y,X,\theta^{*})^{2}]=E[s(X'\gamma^{*})]$, and
\[
Q(\theta^{*})=\frac{1}{2}E[\{e(Y,X,\theta^{*})^{2}+1\}s(X'\gamma^{*})]=E[s(X'\gamma^{*})].
\]
We have shown that for the linear and exponential scale specifications,
$Q(\theta^{*})=E[e(Y,X,\theta^{*})^{2}s(X'\gamma^{*})]$. For OLS,
conditions (\ref{eq:FOC2}) simplify to $E[e(Y,X,\theta_{\textrm{LS}})^{2}-1]=0$,
which implies that $s(\gamma_{\textrm{LS}})=E[(Y-X'\beta_{\textrm{LS}})^{2}]^{1/2}$,
and
\[
Q(\theta_{\textrm{LS}})=\frac{1}{2}E[\{e(Y,X,\theta_{\textrm{LS}})^{2}+1\}s(\gamma_{\textrm{LS}})]=s(\gamma_{\textrm{LS}}).
\]
We have shown that for the constant scale specification $Q(\theta_{\textrm{LS}})=E[(Y-X'\beta_{\textrm{LS}})^{2}]^{1/2}$.
The result then follows by Theorem \ref{thm:Interpretation}(ii).\qed

\subsection{Proof of Theorem \ref{thm:MLequivalence}\label{subsec:Proof-of-TheoremKLIC}}

By Theorem \ref{thm:Interpretation}(ii) we have that $E\left[s(X'\gamma^{*})\right]=Q(\theta^{*})\leq Q(\theta_{\textrm{LS}})=\sigma_{\textrm{LS}}$.
This and Jensen's inequality for concave functions imply that
\begin{equation}
E\left[\log s(X'\gamma^{*})\right]\leq\log E\left[s(X'\gamma^{*})\right]\leq\log\left(\sigma_{\textrm{LS}}\right),\label{eq:ObjveIneq}
\end{equation}
by monotonicity of the logarithmic function. Thus
\[
\epsilon=2\left\{ \log\left(\sigma_{\textrm{LS}}\right)-E\left[\log s(X'\gamma^{*})\right]\right\} \geq0.
\]
By definition (\ref{eq:Gausdens}) of the Gaussian density, inequality
(\ref{eq:ObjveIneq}), and $E[e\left(Y,X,\theta_{\textrm{LS}}\right)^{2}]=1$,
the bound $E\left[e(Y,X,\theta^{*})^{2}\right]\leq1+\epsilon$, implies
\begin{eqnarray}
-E\left[\log f_{\theta^{*}}(Y,X)\right] & = & \frac{1}{2}\log\left(2\pi\right)+E\left[\log s(X'\gamma^{*})+\frac{1}{2}e(Y,X,\theta^{*})^{2}\right]\nonumber \\
 & \leq & \frac{1}{2}\left[\log\left(2\pi\right)+1\right]+\log\left(\sigma_{\textrm{LS}}\right)\nonumber \\
 & = & -E\left[\log f_{\theta_{\textrm{LS}}}(Y,X)\right].\label{eq:KLineq1}
\end{eqnarray}
Therefore, upon adding $E[\log f_{Y\mid X}(Y\mid X)]$ to each side
of the inequality, we obtain inequality (\ref{eq:-1}).\qed

\subsection{Proof of Corollary \ref{cor:CLM}}

By assumption $\mu(X)=X'\beta_{0}$ a.s., and by definition $\theta^{*}$
satisfies conditions (\ref{eq:FOCmiss1})-(\ref{eq:FOCmiss2}). Then
the first-order conditions (\ref{eq:FOCmiss1}) are satisfied by $\beta^{*}=\beta_{0}$,
and $\gamma^{*}$ must satisfy
\begin{equation}
E\left[Xs_{1}(X'\gamma^{*})\left\{ \frac{\sigma(X)^{2}}{s(X'\gamma^{*})^{2}}-1\right\} \right]=0.\label{eq:FOCgamma*}
\end{equation}
If there exists a pair $(\beta,\gamma)\in\Theta$ satisfying all $2\times k$
conditions (\ref{eq:FOCmiss1})-(\ref{eq:FOCmiss2}) simultaneously,
then this pair is unique, by strict convexity of $\theta\mapsto Q(\theta)$.
It follows from the existence proof of Theorem \ref{thm:Existence}
that the restriction $q(\gamma):=Q(\theta)|_{\beta=\beta_{0}}$ has
a minimizer in $\Theta_{\gamma}$, i.e., there exists $\gamma^{*}$
such that (\ref{eq:FOCgamma*}) holds. Since $\gamma\mapsto q(\gamma)$
is also strictly convex, $q(\gamma)$ admits a unique minimizer $\gamma^{*}(\beta_{0})$
in $\Theta_{\gamma}$. Therefore, the pair $(\beta_{0},\gamma^{*}(\beta_{0}))$
is the unique minimizer of $Q(\theta)$ when $\mu(X)=X'\beta_{0}$
a.s.\qed

\subsection{Proof of Theorem \ref{thm:MLequivalence2}\label{subsec:Proof-of-TheoremKLIC-2}}

Letting $L^{1}(\mathcal{X})=\{\psi:\int\left|\psi(x)\right|f_{X}(x)dx<\infty\}$,
define the set of admissible conditional standard deviation functions
as
\[
\mathcal{S}=\left\{ \psi\in L^{1}(\mathcal{X})\,:\,\Pr[\psi(X)>0]=1\right\} .
\]

We first show that $\sigma(X)$ is a minimizer of the MVR criterion
\[
Q_{0}(\psi):=E\left[\frac{1}{2}\left\{ \left(\frac{Y-X'\beta_{0}}{\psi(X)}\right)^{2}+1\right\} \psi(X)\right],\qquad\psi\in\mathcal{S}.
\]
Assuming that the order of differentiation and integration can be
interchanged, as in Lemma \ref{lem:QDiff}, and setting $\psi_{\alpha}(X):=\sigma(X)+\alpha\psi(X)$
and $e_{\alpha}(Y,X):=(Y-X'\beta_{0})/\psi_{\alpha}(X)$, we have
\begin{align*}
\left.\frac{\partial Q_{0}(\psi_{\alpha})}{\partial\alpha}\right|_{\alpha=0} & =\frac{1}{2}\left.E\left[2e_{\alpha}(Y,X)\frac{\partial e_{\alpha}(Y,X)}{\partial\alpha}\psi_{\alpha}(X)+\frac{\partial\psi_{\alpha}(X)}{\partial\alpha}e_{\alpha}(Y,X)^{2}+\frac{\partial\psi_{\alpha}(X)}{\partial\alpha}\right]\right|_{\alpha=0}\\
 & =\frac{1}{2}\left.E\left[-\frac{\partial\psi_{\alpha}(X)}{\partial\alpha}\left\{ e_{\alpha}(Y,X)^{2}-1\right\} \right]\right|_{\alpha=0},
\end{align*}
upon using that $\partial e_{\alpha}(Y,X)/\partial\alpha=-e_{\alpha}(Y,X)\{\partial\psi_{\alpha}(X)/\partial\alpha\}\{\psi_{\alpha}(X)\}^{-1}$.
Therefore
\[
\left.\frac{\partial Q_{0}(\psi_{\alpha})}{\partial\alpha}\right|_{\alpha=0}=\frac{1}{2}E\left[-\psi(X)\{e(Y,X)^{2}-1\}\right]=0,\qquad e(Y,X):=\frac{Y-X'\beta_{0}}{\sigma(X)},
\]
for each $\psi\in\mathcal{S}$ since $E\left[e(Y,X)^{2}\mid X\right]=1$
implies that $E\left[h(X)\left\{ e(Y,X)^{2}-1\right\} \right]=0$
for any measurable function of $X$, and in particular each $\psi\in\mathcal{S}$.

We next show that $Q_{0}(\psi)$ is strictly convex and thus admits
at most one minimizer. Assuming that the order of differentiation
and integration can be interchanged, as in Lemma \ref{lem:Convexity},
we have
\begin{align*}
\left.\frac{\partial^{2}Q_{0}(\psi_{\alpha})}{\partial\alpha^{2}}\right|_{\alpha=0} & =\frac{1}{2}\left.E\left[-2\psi(X)e_{\alpha}(Y,X)\frac{\partial e_{\alpha}(Y,X)}{\partial\alpha}\right]\right|_{\alpha=0}\\
 & =\left.E\left[\frac{\psi(X)^{2}}{\psi_{\alpha}(X)}e_{\alpha}(Y,X)^{2}\right]\right|_{\alpha=0}\\
 & =E\left[\frac{\psi(X)^{2}}{\sigma(X)}e(Y,X)^{2}\right]\\
 & =E\left[\frac{\psi(X)^{2}}{\sigma(X)}\right]>0,
\end{align*}
for each $\psi\in\mathcal{S}$. Therefore $\sigma(X)$ is the unique
minimizer of $Q_{0}(\psi)$ in $\mathcal{S}$.

Under correct specification of the CMF, we have that $\beta^{*}=\beta_{0}$,
by Corollary \ref{cor:CLM}. By $s(X'\gamma^{*})\in\mathcal{S}$ and
$\sigma(X)$ being the unique minimizer of $Q_{0}(\psi)$ in $\mathcal{S}$
we have that $Q_{0}(\sigma)\leq Q(\theta^{*})$. By $Q_{0}(\sigma)=E\left[\sigma(X)\right]$
and $Q(\theta^{*})=E\left[s(X'\gamma^{*})\right]$, we have shown
that $E\left[\sigma(X)\right]\leq E\left[s(X'\gamma^{*})\right]$.
This and Jensen's inequality for concave functions imply
\begin{equation}
E\left[\log\sigma(X)\right]\leq\log E\left[\sigma(X)\right]\leq\log E\left[s(X'\gamma^{*})\right],\label{eq:ObjveIneq-1}
\end{equation}
by monotonicity of the logarithmic function.

By iterated expectations and using that $z\mapsto z-1-\log z$ is
a positive-valued strictly convex function for $z\geq0$ gives the
bound
\begin{equation}
E\left[e(Y,X,\theta^{*})^{2}\right]=E\left[\left(\frac{\sigma(X)}{s(X'\gamma^{*})}\right)^{2}\right]\geq1+2\left\{ E\left[\log\sigma(X)\right]-E\left[\log s(X'\gamma^{*})\right]\right\} .\label{eq:Bound-2}
\end{equation}
By definition of the Gaussian density, inequality (\ref{eq:ObjveIneq-1}),
and $E[e\left(Y,X\right)^{2}]=1$, bound (\ref{eq:Bound-2}) implies
\begin{eqnarray}
-E\left[\log f_{\theta^{*}}(Y,X)\right] & = & \frac{1}{2}\log\left(2\pi\right)+E\left[\log s(X'\gamma^{*})+\frac{1}{2}e(Y,X,\theta^{*})^{2}\right]\nonumber \\
 & \geq & \frac{1}{2}\left[\log\left(2\pi\right)+1\right]+E\left[\log\sigma(X)\right]\nonumber \\
 & = & -E\left[\log f_{\beta_{0}}^{\dagger}(Y,X)\right].\label{eq:KLineq2}
\end{eqnarray}
Upon combining (\ref{eq:KLineq1}) and (\ref{eq:KLineq2}) we obtain
inequality (\ref{eq:-2}).\qed

\section{Asymptotic Theory}
\begin{lem}
\label{thm:posdefn}Suppose that Assumptions \ref{ass:ScaleSpec},
\ref{ass:Moments} and \ref{ass:Sample} holds. Then $Q_{n}(\theta)$
is strictly convex over $\mathcal{B}$.
\end{lem}
\begin{proof}
For $e_{i}=e(y_{i},x_{i},\theta)$, the Hessian matrix $H_{n}(\theta)$
of $Q_{n}(\theta)$,
\begin{eqnarray*}
H_{n}(\theta) & = & \frac{1}{n}\sum_{i=1}^{n}\left[\begin{array}{cc}
\frac{x_{i}x_{i}'}{s(x_{i}'\gamma)} & \frac{x_{i}x_{i}'}{s(x_{i}'\gamma)}s_{1}(x_{i}'\gamma)e_{i}\\
\frac{x_{i}x_{i}'}{s(x_{i}'\gamma)}s_{1}(x_{i}'\gamma)e_{i} & \frac{x_{i}x_{i}'}{s(x_{i}'\gamma)}\{s_{1}(x_{i}'\gamma)e_{i}\}^{2}
\end{array}\right]\\
 &  & +\frac{1}{n}\sum_{i=1}^{n}\left[\begin{array}{cc}
0_{k\times k} & 0_{k\times k}\\
0_{k\times k} & -\frac{1}{2}x_{i}x_{i}'s_{2}(x_{i}'\gamma)(e_{i}^{2}-1)
\end{array}\right]:=H_{1n}(\theta)+H_{2n}(\theta),
\end{eqnarray*}
defined for $\theta\in\mathcal{B}$, is positive definite. Steps similar
to the proof of Lemma \ref{thm:posdef} show that $H_{1n}(\theta)$
is positive definite for all $\theta\in\mathcal{B}$. Moreover, all
principal minors of $H_{2n}(\theta)$ have determinant $0$ for all
$\theta\in\mathcal{B}$, and $H_{2n}(\theta)$ is thus positive semidefinite.
Since $H_{n}(\theta)=H_{1n}(\theta)+H_{2n}(\theta)$, we conclude
that $H_{n}$ is positive definite for all $\theta\in\mathcal{B}$
(\citealp{HJ:2012}, p.398, 7.1.3. observation), and the result follows.
\end{proof}

\subsection{Proof of Theorem \ref{thm:ANorm}}

\textit{Proof of part (i)-(ii). }By Theorem \ref{thm:Existence},
$\theta^{*}\in\Theta$ is the unique minimizer of $Q(\theta)$, and
the identification condition (i) in Theorem 2.7 in \citet{Newey:McFadden:1994}
is thus verified. Since $\Theta$ is convex and open, existence of
$\theta^{*}\in\Theta$ established in Theorem \ref{thm:Existence},
as well as strict convexity of $Q_{n}(\theta)$ established in Lemma
\ref{thm:posdefn} imply that their condition (ii) is satisfied. Finally,
since the sample is i.i.d. by Assumption \ref{ass:Sample}, pointwise
convergence of $Q_{n}(\theta)$ to $Q_{0}(\theta)$ follows from $Q_{0}(\theta)$
bounded (established in the proof of Lemma \ref{lem:QCont}) and application
of Khinchine's law of large numbers. Hence, all conditions of Newey
and McFadden's Theorem 2.7 are satisfied. Therefore, there exists
$\hat{\theta}\in\Theta$ with probability approaching one, and $\hat{\theta}\rightarrow^{p}\theta^{*}$.\qed

\textit{Proof of part (iii)}. The sample MVR solution $\hat{\theta}$
can be equivalently formulated as the Method-of-Moments estimator
\[
\hat{\theta}=\arg\min_{\theta\in\varTheta}\left[\frac{1}{n}\sum_{i=1}^{n}m(y_{i},x_{i},\theta)\right]'\left[\frac{1}{n}\sum_{i=1}^{n}m(y_{i},x_{i},\theta)\right],
\]
The asymptotic normality result $n^{1/2}(\hat{\theta}-\theta^{*})\overset{d}{\rightarrow}N(0,G^{-1}S(G^{-1})')$
then follows from this characterization upon verifying the assumptions
of Theorem 3.4 in \citet{Newey:McFadden:1994}, for instance. Block
symmetry of $G$ then implies that $V=G^{-1}SG^{-1}$.

By Theorem \ref{thm:Existence}, $\theta^{*}$ is in the interior
of $\Theta$ so that their Condition (i) is satisfied. The mapping
$\theta\mapsto m(Y,X,\theta)$ is continuously differentiable, by
inspection, so that their Condition (ii) is satisfied. By definition,
$\theta^{*}$ satisfies $E[m(Y,X,\theta^{*})]=0$, hence the first
part of their condition (iii) is satisfied. Moreover, bound (\ref{eq:BoundExp})
in the proof of Lemma \ref{lem:QDiff} and the steps below show that
$E[||m(Y,X,\theta^{*})||^{2}]$ is finite under Assumption \ref{ass:ANorm},
verifying their Condition (iii). Finally, under our assumptions, from
the proof of Lemma \ref{lem:Convexity}, $E[\sup_{\theta\in\Theta}||\partial m(Y,X,\theta)/\partial\theta||]=E[\sup_{\theta\in\Theta}||\partial^{2}L(X,Y,\theta)/\partial\theta\partial\theta||]$
is finite and $G=E[\partial m(Y,X,\theta^{*})/\partial\theta]$ is
nonsingular. Their Conditions (iv) and (v) are satisfied.

If $\mu(X)=X'\beta^{*}$ a.s., then $E[(Y-X'\beta^{*})/s(X'\gamma^{*})\mid X]=0$.
Therefore, by iterated expectations, the off-diagonal blocks of $G$
and $S$ simplify to (\ref{eq:Var_miss}). If $\mu(X)=X'\beta^{*}$
a.s. and $\sigma^{2}(X)=s(X'\gamma^{*})^{2}$ a.s., then $E[e^{2}-1\mid X]=0$.
Therefore, repeated use of iterated expectations imply
\[
E\left[\frac{XX'}{s(X'\gamma^{*})}\{s_{1}(X'\gamma^{*})e\}^{2}\right]-E\left[\frac{1}{2}XX's_{2}(X'\gamma^{*})(e{}^{2}-1)\right]=E\left[\frac{XX'}{s(X'\gamma^{*})}s_{1}(X'\gamma^{*})^{2}\right],
\]
which yields $G_{22}$ in (\ref{eq:Correct_spec}), and $S_{11}$
and $S_{22}$ in (\ref{eq:Correct_spec}).

\begin{sloppy}Application of Theorem 4.5 in \citet{Newey:McFadden:1994}
implies that $\hat{G}^{-1}\hat{S}\hat{G}^{-1}\rightarrow^{p}G^{-1}SG^{-1}$.
Under Assumption \ref{ass:ANorm}, steps similar to the proof of Lemma
\ref{lem:QDiff} show that for a neighborhood $\mathcal{N}$ of $\theta^{*}$
we have $E[\sup_{\theta\in\mathcal{N}}||m(Y,X,\theta)||^{2}]<\infty$:
bound (\ref{eq:BoundExp}) implies, for all $\theta\in\mathcal{N}$,
\begin{align*}
\left\Vert Xs_{1}(X'\gamma)\{e(Y,X,\theta)^{2}-1\}\right\Vert ^{2} & \leq\left\{ C[||X||+||X||s_{2}(X'\bar{\gamma})][Y^{2}+||X||^{2}]\right\} ^{2}\\
 & \leq2C[||X||^{2}+||X||^{2}s_{2}(X'\bar{\gamma})^{2}][Y^{4}+||X||^{4}],
\end{align*}
and $E[\sup_{\theta\in\mathcal{N}}||m(Y,X,\theta)||^{2}]<\infty$
requires $E[Y^{4}||X||^{2}]<\infty$, $E[||X||^{6}]<\infty$, and,
for all $\gamma\in\Theta_{\gamma}$, $E[Y^{4}||X||^{2}s_{2}(X'\gamma)^{2}]<\infty$,
and $E[||X||^{6}s_{2}(X'\gamma)^{2}]<\infty$. By Holder's inequality
$E[Y^{4}||X||^{2}]<\infty$ and $E[Y^{4}||X||^{2}s_{2}(X'\gamma)^{2}]<\infty$
if $E[Y^{6}]<\infty$, $E[||X||^{6}]<\infty$, and $E[||X||^{6}s_{2}(X'\gamma)^{6}]<\infty$,
which hold under Assumption \ref{ass:ANorm}. Moreover, $\theta\mapsto m(Y,X,\theta)$
is continuous at $\theta^{*}$ a.s.  The result follows.\qed\par\end{sloppy}

\end{document}